\documentclass[accepted]{uai2025} 
                        
\usepackage[dvipsnames]{xcolor}
\usepackage[american]{babel}

\usepackage{natbib} 
\usepackage{mathtools} 
\usepackage{booktabs} 
\usepackage{tikz} 
\usepackage{caption}
\usepackage{subcaption} 
\usepackage{amsfonts} 
\usepackage{amsthm}
\usepackage{thmtools}
\usetikzlibrary{arrows.meta} 
\usepackage{algorithm}
\usepackage{algpseudocode}
\usepackage{array}
\usepackage{multirow}
\usepackage{csquotes}
\usepackage{amssymb}
\usepackage{graphbox}

\theoremstyle{plain}

\newtheorem{lemma}{Lemma}[section]
\newtheorem{corollary}{Corollary}[section]

\newtheorem{assumption}{Assumption}[section]

\theoremstyle{definition}

\newtheorem{example}{Example}
\theoremstyle{remark}
\newtheorem{remark}{Remark}[section]

\usepackage[capitalize]{cleveref}

\Crefname{assumption}{Assumption}{Assumptions}
\crefname{assumption}{Assumption}{Assumptions}
\Crefname{condition}{Condition}{Conditions}
\crefname{condition}{Condition}{Conditions}
\Crefname{secinapp}{Appendix}{Appendices}
\crefname{secinapp}{Appendix}{Appendices}

\renewcommand{\thmcontinues}[1]{%
\ifcsname hyperref\endcsname
    \hyperref[#1]{continued}%
\else
    continued%
\fi%
}

\definecolor{friendlygreen}{RGB}{17, 119, 51}
\definecolor{friendlyolive}{RGB}{153, 153, 51}
\definecolor{friendlypurple}{RGB}{170, 68, 153}
\definecolor{friendlyindigo}{RGB}{51, 34, 136}
\definecolor{friendlyblue}{RGB}{0, 119, 187}
\newcommand{\colora}[1]{{\color{friendlygreen}#1}}
\newcommand{\colorb}[1]{{\color{friendlypurple}#1}}
\newcommand{\colorc}[1]{{\color{friendlyblue}#1}}

\DeclareMathOperator{\circlearrow}{\hbox{$\circ$}\kern-1.5pt\hbox{$\rightarrow$}}
\DeclareMathOperator{\circlecircle}{\hbox{$\circ$}\kern-1.2pt\hbox{$--$}\kern-1.5pt\hbox{$\circ$}}

\DeclareMathOperator*{\argmax}{arg\,max}
\DeclareMathOperator*{\argmin}{arg\,min}

\def\ci{\perp\!\!\!\perp}

\newcommand{\E}{\mathbb{E}}

\newcommand{\R}{\mathbb{R}}
\newcommand{\cX}{\mathcal{X}}
\newcommand{\cM}{\mathcal{M}}
\newcommand{\cD}{\mathcal{D}}
\newcommand{\cS}{\mathcal{S}}
\newcommand{\cA}{\mathcal{A}}

\newcommand{\1}[1]{\mathbf{1}\{#1\}}
\newcommand{\Pe}{\mathbf{\Pi}^{e}}
\newcommand{\Pel}{\mathbf{\Pi}^{e}_{\leq}}
\newcommand{\Peg}{\mathbf{\Pi}^{e}_{\geq}}
\newcommand{\Petl}{\tilde{\mathbf{\Pi}}^{e}_{\leq}}
\newcommand{\Petg}{\tilde{\mathbf{\Pi}}^{e}_{\geq}}

\usepackage{hyperref}
\hypersetup{
  colorlinks=true,
  linkcolor=NavyBlue,
  citecolor=gray
}

\graphicspath{{.}{../}}

\title{Just Trial Once: Ongoing Causal Validation of Machine Learning Models}

\author[1]{Jacob M. Chen}
\author[1]{Michael Oberst}
\affil[1]{%
    Department of Computer Science\\
    Johns Hopkins University
}
  
\begin{document}
\maketitle

\begin{abstract}
    Machine learning (ML) models are increasingly used as decision-support tools in high-risk domains. Evaluating the causal impact of deploying such models can be done with a randomized controlled trial (RCT) that randomizes users to ML vs. control groups and assesses the effect on relevant outcomes. However, ML models are inevitably updated over time, and we often lack evidence for the causal impact of these updates. While the causal effect could be repeatedly validated with ongoing RCTs, such experiments are expensive and time-consuming to run. In this work, we present an alternative solution: using only data from a prior RCT, we give conditions under which the causal impact of a new ML model can be precisely bounded or estimated, even if it was not included in the RCT. Our assumptions incorporate two realistic constraints: ML predictions are often deterministic, and their impacts depend on user trust in the model. Based on our analysis, we give recommendations for trial designs that maximize our ability to assess future versions of an ML model. Our hope is that our trial design recommendations will save practitioners time and resources while allowing for quicker deployments of updates to ML models.
\end{abstract}

\section{Introduction}
\label{sec:intro}

Machine learning (ML) models are increasingly deployed in high-risk domains like healthcare and criminal justice as tools to support human decision-makers.  For instance, in healthcare, ML-powered decision-support tools (ML-DSTs) are widespread, including early warning systems for sepsis \citep{Adams2022-e7, Sendak2020-c7, Boussina2024-30}, computer-assisted decision-support for antibiotic treatment decisions \citep{Gohil2024-10,Gohil2024-82}, and a variety of tools for computer-aided diagnostics in radiology and pathology, with the FDA having cleared or approved over 1,000 AI/ML-enabled devices to date \citep{FDA2024-c4}. Although these models often exhibit high accuracy, it is not always clear whether their deployment actually leads to better decisions, and thus, better downstream outcomes.  In healthcare, for instance, we are interested not only in model accuracy, but also whether deployment of an ML-DST improves health outcomes for patients. 

The gold standard evaluation of ML-enabled decision-support is to assess impact in a randomized controlled trial (RCT), typically structured as a cluster RCT, where decision-makers (e.g., clinicians in a given hospital) are randomized to an ML-DST or no ML-DST. Such trials are becoming more common in healthcare~\citep{Han2024-0b} and criminal justice~\citep{Imai2023-04}. Examples include recent \enquote{failed trials} like the PROTEUS trial of ML-assisted diagnosis of stress echocardiography~\citep{Upton2024-aa} and trials with more positive results, such as the INSPIRE trials for antibiotic recommendations powered by ML predictions of resistance likelihood~\citep{Gohil2024-10,Gohil2024-82}. These trials provide rigorous evidence for the impact of deploying specific ML-enabled systems (and their underlying models), and the broader research community recognizes the need for more randomized trials~\citep{ouyang2024we} and evaluation of ML systems as interventions~\citep{joshi2025ai}.

However, the traditional RCT framework is not designed for ML-enabled systems, which (unlike drugs) are often updated frequently to handle performance degradation. Even when RCT data is available for a single version of an ML-DST, it is not obvious whether those results apply to later models, and running additional RCTs to verify continued effectiveness is both time-consuming and costly.

\begin{figure*}[ht]
    \begin{center}
        \begin{subfigure}[t]{0.32\textwidth}
            \centering
            \includegraphics[align=c,bmargin=6pt,scale=0.25]{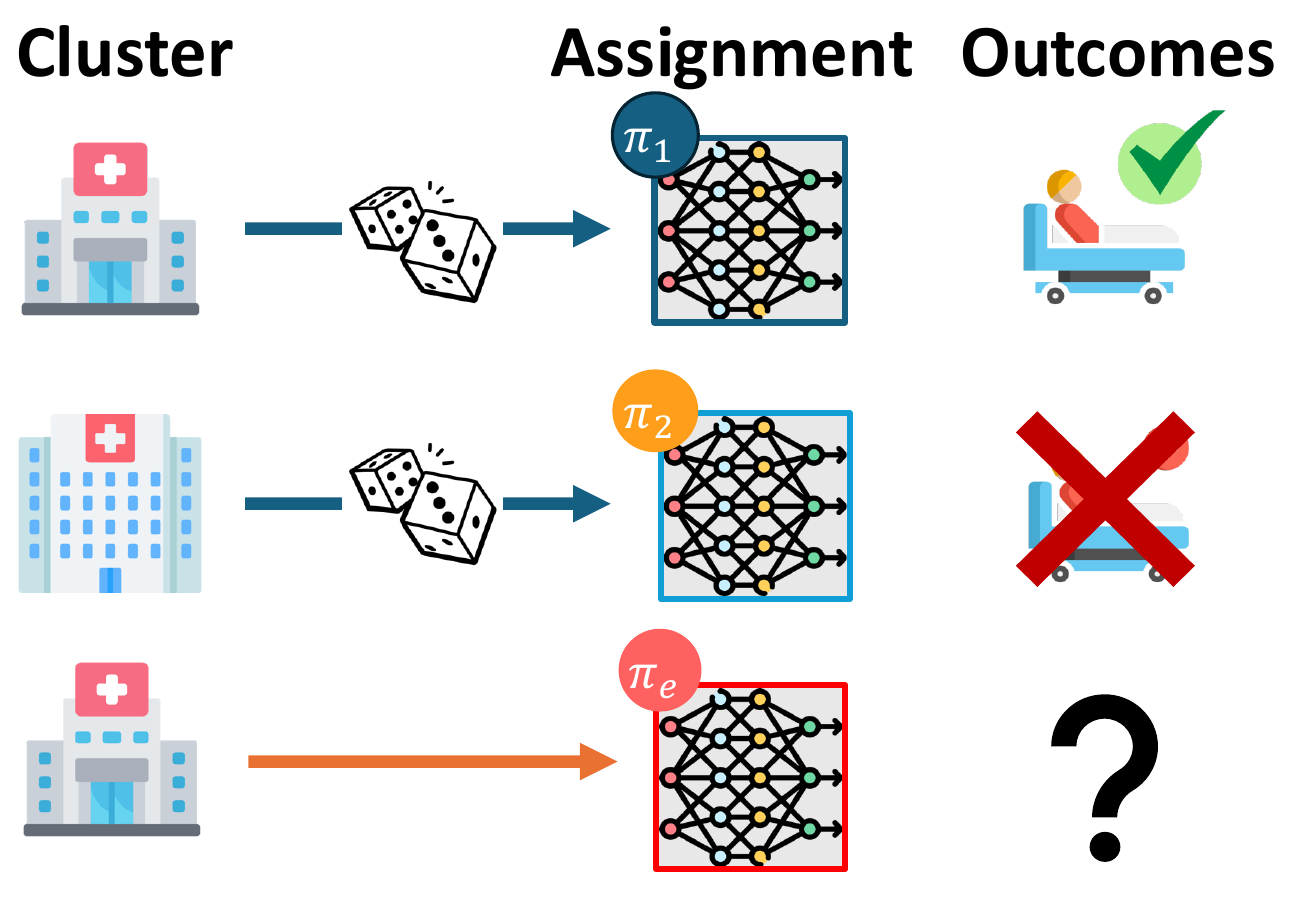}
            \caption{}
            \label{fig:goal}
        \end{subfigure}
        \begin{subfigure}[t]{0.32\textwidth}
            \centering
            \includegraphics[align=c,bmargin=6pt,scale=0.22]{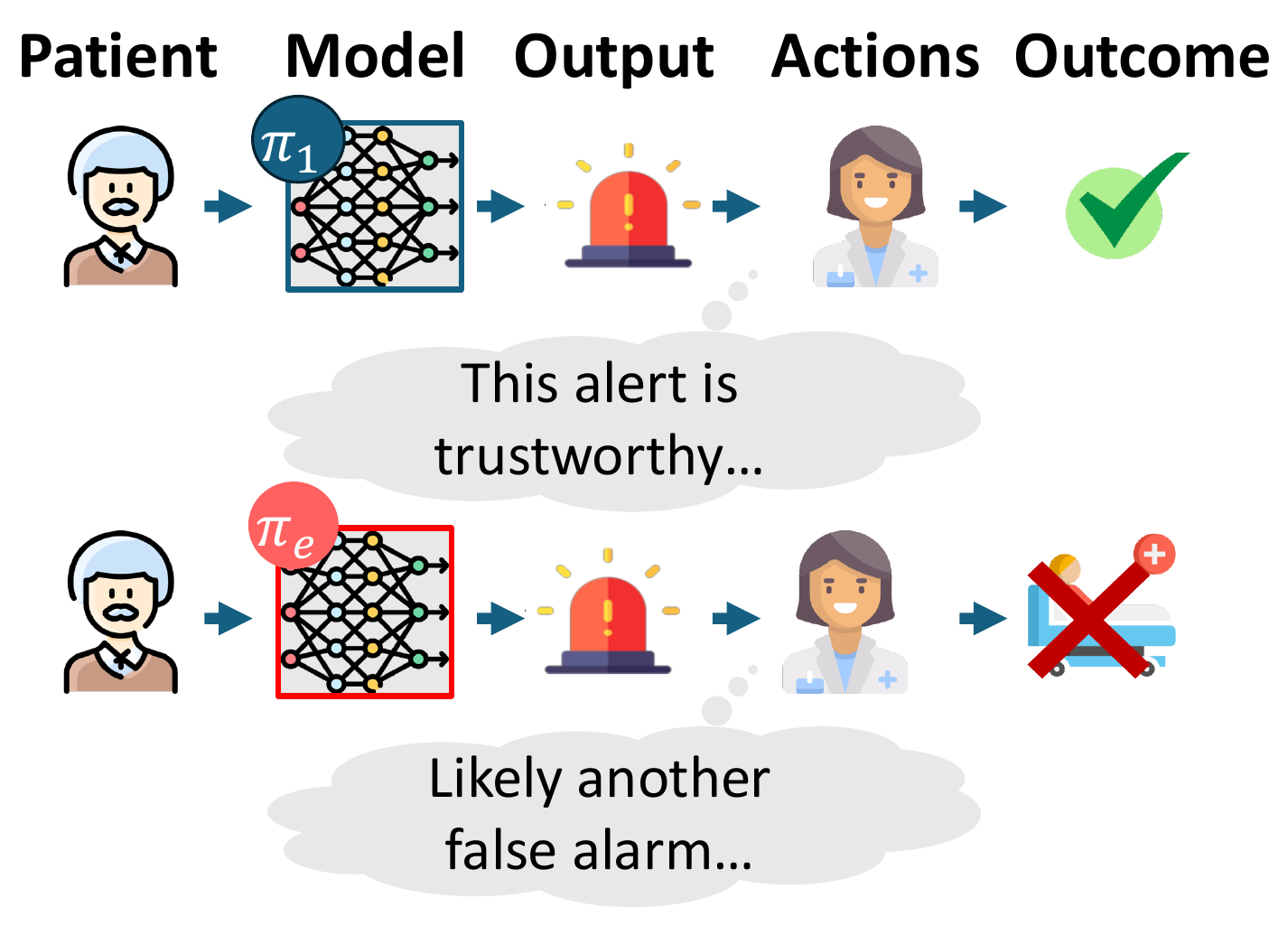}
            \caption{}
            \label{fig:performance}
        \end{subfigure}
        \begin{subfigure}[t]{0.32\textwidth}
            \centering
            \includegraphics[align=c,scale=0.32]{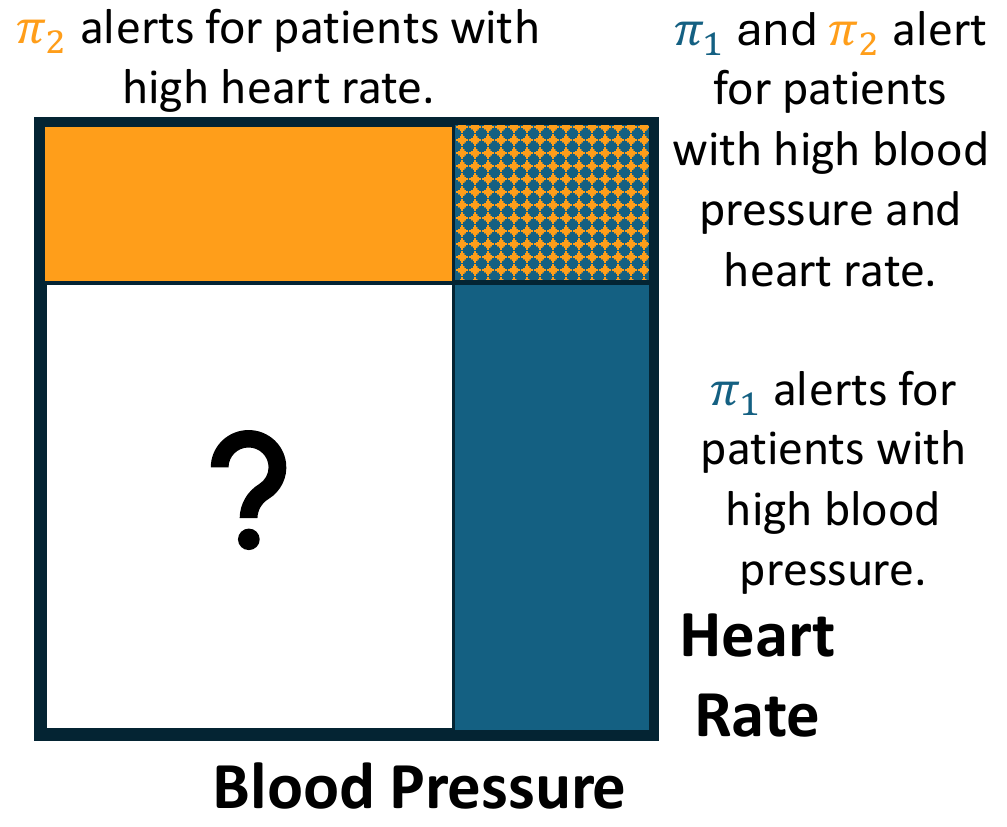}
            \caption{}
            \label{fig:overlap}
        \end{subfigure}
    \end{center}
    \caption{(\subref{fig:goal}) The goal of this paper is to predict the causal impact of deploying a new model $\pi_e$, given data from a cluster randomized trial that randomizes sets of users (e.g., hospitals) to one of a set of trial models that does not include $\pi_e$. 
    (\subref{fig:performance}) The first challenge: Relevant outcomes (e.g., of patients) are not only influenced by model outputs, but also by how users actually respond to the outputs of model-based decision-support, which may itself be affected by the perceived reliability / performance of the model.
    (\subref{fig:overlap}) The second challenge: There may exist some subset of cases (e.g., patients) for whom we never observe certain model outputs, making it impossible to give precise predictions for outcomes of patients in that group. In this example, $\pi_1$ alerts for patients with high blood pressure, and $\pi_2$ alerts for patients with high heart rate; thus, patients with high blood pressure and high heart rate receive an alert from both $\pi_1$ and $\pi_2$. However, patients in the white region -- those who have both low to medium blood pressure and heart rate -- never receive an alert from either trialed model.}
    \label{fig:conceptual_illustration}
\end{figure*}

Our work addresses this challenge from a methodological perspective, as illustrated in~\cref{fig:conceptual_illustration}:  We formalize conditions under which data from an existing RCT can be used to precisely infer or bound the causal impact of deploying models that were not included in the original RCT. We take into account two important practical considerations: First, model performance (e.g., accuracy at diagnosing disease) will influence user trust in the system, and thereby indirectly influence outcomes (\cref{fig:performance}). Second, while the deployment of DSTs is often randomized, the predictions themselves are not typically randomized (\cref{fig:overlap}), since doing so would undermine trust (e.g., by raising alerts randomly).  Hence, there may be some combinations of model outputs (e.g., diagnoses) and inputs (e.g., patients) that we never observe.

Under limited assumptions that incorporate these considerations, we derive bounds on the causal impact of deploying a new model. Crucially, we show that both of our main assumptions can be checked using RCT data that includes at least two models with differing performance characteristics. In a simulation study, we show how our framework yields more rigorous conclusions about the value of model updates, as compared to naive approaches that only judge models based on their raw performance.

Our results have practical implications for post-trial analysis and pre-trial design. First, evaluating new models using historical trial data is possible under reasonably limited assumptions, but not all alternative models can be precisely evaluated in this way. Second, our results suggest a benefit to running RCTs with multiple ML models to maximize the ability to estimate causal impacts in future model updates. 

To summarize, our contributions are as follows:
\begin{itemize}
    \item We provide assumptions (\cref{asmp:scm,asmp:more_reliable_better,asmp:A_control_value,asmp:bounded_Y}) under which we derive bounds (\cref{thm:bounds_with_control}) on the effect of deploying a new ML model, given data from a prior RCT. Our bounds are tight, and cannot be improved without further assumptions (\cref{prop:non-id}). 
    \item We provide a simple estimator for these bounds and a procedure for generating asymptotically valid confidence intervals (\cref{cor:empirical_estimation}). We also show that our core assumptions can be falsified via hypothesis tests constructed from RCT data trialing multiple models (\cref{prop:falsify_monotonicity,prop:falsify_neutral_actions}).
    \item We provide recommendations for pre-trial design and post-trial analysis in light of our results (\cref{sec:implications}), and demonstrate in a simulation study (\cref{sec:simulation}) that our bounds provide a more informative tool to select among model updates as compared to using the raw performance (e.g., accuracy) of updated models.
\end{itemize}

{\bf Related literature:} 
Our work is related to \textit{off-policy policy evaluation} in causal inference and reinforcement learning~\citep{uehara2022review}. An ML-DST can be viewed as a deterministic policy that chooses actions (i.e., predictions or alerts to raise) based on context (i.e., inputs to the model) with the goal of obtaining some reward (i.e., positively influencing outcomes for patients). Two critical distinctions arise in our work versus the standard setting: First, the policies that are present in retrospective data (in our case, from a trial) are deterministic rather than random, leading to violation of the common assumption that, for a given context, there is a positive probability of seeing any action. Second, we allow for the fact that actions taken for one patient can influence outcomes for other patients.

Our work is also related to causal evaluations of AI-assisted decision-making in criminal justice settings \citep{Imai2023-04,ben2024does,ben2025safe}, but our goal differs: Rather than evaluating the impact of AI-assistance on the accuracy of (observable) predictions made by a human, we are interested in the total effect of model deployment on downstream outcomes. Finally, our work is connected to the study of causal transportability where the goal is to infer the effect of a known intervention from an RCT onto a new target population where randomization is difficult, expensive, or impossible \citep{pearl2011transportability, stuart2011use}. This is similar to our setting as we are also attempting to draw inference from RCT data. However, our problem differs in that we would like to infer the effect of an unseen intervention (i.e., a new model with no historical trial data) on the same population as in the original RCT.

\section{Model and Problem Setup}\label{sec:model_and_setup}

\textbf{Notation}: In the rest of this paper, we use the terms {\it model} and {\it policy} interchangeably. We use upper case letters $X$ to denote a random variable, calligraphic font $\cX$ to denote the space of possible values, and lower-case letters $x$ to denote a specific realization of a random variable. We assume that the causal structure of an RCT is modeled by a directed acyclic graph (DAG) $\mathcal{G}$ over a set of vertices ${\bf V} = \{A, Y, D, X, \Pi, M\}$, where 
$A \in \cA$ represents the output (or \enquote{action}) of the deployed model, 
$Y \in \R$ represents an outcome of interest,
$D \in \cD$ represents the cluster to which a user is assigned, 
$X \in \cX$ represents covariates used as inputs to the ML model, 
$\Pi \in \varPi$ represents the specific ML model that was deployed, and 
$M \in \R$ represents model performance, which we represent as a real number. We assume that model performance is computable for any model via some functional $f_M(\pi)$ (e.g., the accuracy, precision, recall, sensitivity, specificity, or some combination, computed on a held-out dataset where $\pi(X)$ is considered the model prediction). We also use the indicator function $\1{S}$ that is equal to $1$ if the event $S$ is true, and $0$ otherwise.
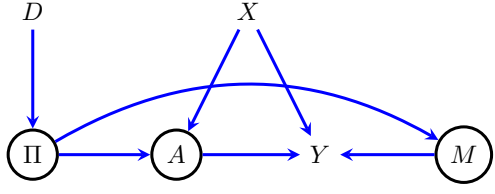
\begin{figure}[t]
    \centering
    \scalebox{0.95}{
        \begin{tikzpicture}[>=stealth, node distance=2cm]
            \tikzstyle{square} = [draw, thick, minimum size=1.0mm, inner sep=3pt]
            \begin{scope}
                \path[->, very thick]
                node[circle, draw] (a) {$A$}
                node[right of=a] (y) {$Y$}
                node[above of=a, xshift=1cm] (x) {$X$}
                node[right of=y, circle, draw] (m) {$M$}
                node[left of=a, circle, draw] (pi) {$\Pi$}
                node[above of=pi] (d) {$D$}
                
                (a) edge[blue] (y)
                (x) edge[blue] (a)
                (x) edge[blue] (y)
                (d) edge[blue] (pi)
                (pi) edge[blue] (a)
                (m) edge[blue] (y)
                (pi) edge[blue, bend left] (m)
                ;
            \end{scope}
        \end{tikzpicture}
    }
    \caption{The directed acyclic graph (DAG) $\mathcal{G}$ depicting the causal relationships in our problem setup (\cref{asmp:scm}). We draw circles around the nodes $\Pi, A,$ and $M$ to represent that these variables are deterministic given their parents (the nodes that have a direct edge to them).}
    \label{fig:dag_setup}
\end{figure}

\setcounter{example}{0}
\begin{example}[Alerting Systems]\label{exp:motivating_example}
    Suppose we are interested in the effect of deploying a DST that monitors patient vital signs and selectively raises an \enquote{alert}. A common application of these systems is detecting the onset of sepsis and alerting clinicians to facilitate timely intervention \citep{Adams2022-e7,Sendak2020-c7,Boussina2024-30}. Here, the inputs $X$ to the model are typically vital signs, the outcome $Y$ may be long-term patient survival, and the outputs $\cA$ include raising an alert ($A = 1$) or not $(A = 0)$. The variable $M$ in this setting could correspond to the false alarm rate of the alerting policy $\Pi$ when it comes to predicting the onset of disease within the next hour.  Note that the label used for computing performance here (onset of disease) differs from the patient outcome of interest $Y$ (survival). A control arm of \enquote{no assistance} can be represented as a deterministic policy that never raises an alert.
\end{example} 

\begin{example}[Computer Assisted Diagnosis]\label{exp:motivating_example_2}
    Suppose we are interested in the effect of deploying a diagnostic model that assists with screening for some disease.  Here, the outcome of interest $Y$ may be long-term patient survival, $X$ would include inputs to the model (e.g., medical imaging, past medical history), and the set of actions $\cA$ could include a set of $K$ possible diagnostic labels as well as the option of deferring to a human expert, such that $\cA = \{ \varnothing, 1, \ldots, K\}$, where $\varnothing$ denotes deferral. The variable $M$ in this setting could represent the overall accuracy of the diagnostic model at predicting some true diagnostic label or some combination of its sensitivity and specificity when it does not defer. In a randomized trial where the control arm consists of \enquote{no assistance}, the resulting \enquote{policy} in the control arm could be viewed as a deterministic policy that always defers.
\end{example}

For concreteness in the remainder of this paper, we will primarily use the language of healthcare applications (e.g., patients, likelihood of disease onset, clinical outcomes, etc). Our assumed causal structure can be represented by the structural causal model (SCM) \citep{pearl2009causality} that we define below, which is consistent with the DAG shown in~\cref{fig:dag_setup}.
\begin{assumption}[Data Generating Process]
    The random variables $D \in \cD$, $\Pi \in \varPi$, $X \in \cX$, $A \in \cA$, and $Y \in \mathbb{R}$ are generated according to the SCM
    \begin{align*}
        D &= f_D(\epsilon_D), & X &= f_X(\epsilon_X), \\
        \Pi &= \pi_D, & M &= f_M(\Pi), \\
        A &= \Pi(X) & Y &= f_Y(A, X, M, \epsilon_Y),
    \end{align*}
    where $\epsilon_Y, \epsilon_D,$ and $\epsilon_X$ are mutually independent.
    \label{asmp:scm}
\end{assumption}

We make a few notes regarding~\cref{asmp:scm}. First, the randomization into a specific policy (signified by $D$) is independent of covariates $X$.
Second, the policy $\Pi$ is entirely determined by $D$, model performance $M$ is entirely determined by $\Pi$ (and observable), and the output $A$ is a deterministic function of $X$, based on $\Pi$.
This deterministic nature of model outputs can create difficulties in evaluating new models; in particular, we are unlikely to see all possible outputs $a \in \cA$ for all types of patients $X$.
Finally, we assume that outcomes $Y$ are not only a function of covariates $X$ and the model output $A$, but also the performance $M$ of the model\footnote{We discuss defining $M$ for the control arm in~\cref{sec:control_arms}.}. Note that we assume that $M$ and $A$ are sufficient to capture the impact of a deployed model on outcomes.

\section{Identification and Bounds}
\label{sec:identification}
\label{sec:control_arms}

\textbf{Goal}: We adopt potential outcomes notation~\citep{richardson2013single} where we use $Y(A = a, M = m) \coloneqq f_Y(a, X, m, \epsilon_Y)$ to denote counterfactual outcomes, representing the value of $Y$ that would be observed if we had taken action $A = a$ with a model whose performance is given by $M = m$\footnote{We defer a more detailed discussion of potential outcomes and other causal inference background to~\cref{app:overview}.}. Using this notation, our goal is to infer expected outcomes if we had deployed a new model / policy $\pi_e$ not trialed in the original RCT, i.e. 
\begin{equation}\label{eq:target_estimand}
 \E[Y(\pi_e)] = \E[Y(A=\pi_e, M=f_M(\pi_e))]\footnote{In the rest of this paper, we use $\pi_e$ as shorthand for $\pi_e(X)$.}
\end{equation}
We refer to $\E[Y(A=\pi_e, M=f_M(\pi_e))]$ as our {\it target estimand} or {\it policy value}. Once this value is inferred, one could compute the causal effect of deploying $\pi_e$ as opposed to any other trialed model $\pi_i$ by evaluating $\E[Y(A=\pi_e, M=f_M(\pi_e))] - \E[Y(A=\pi_i, M=f_M(\pi_i))]$.

\begin{example}[continues=exp:motivating_example]
    Suppose that the trialed model alerts based on thresholding a pre-defined risk score $r(x)$ that is a function of vital signs (e.g., systolic blood pressure, respiratory rate, etc). Suppose that an initial RCT assigns patients to a control arm, $D=0$ where $\pi_0 = 0$ (alerts are never raised), and a treatment arm, $D=1$ where the model raises alerts using the threshold $T^*$, i.e., $\pi_1(x) \coloneqq \1{r(x) > T^*}$. Suppose we want to use this RCT data to evaluate the impact of an alternative model with a lower threshold, $\pi_l(x) \coloneqq \1{r(x) > T^l}$ where $T^l < T^*$. \cref{fig:illustration} visually demonstrates the challenges of this inference task for $\pi_l$ as we never observe alerts for patients with $r(x) \in [T^l, T^*]$.
\end{example}

\begin{figure}[t]
    \centering
    \scalebox{1}{
        \includegraphics[scale=0.28]{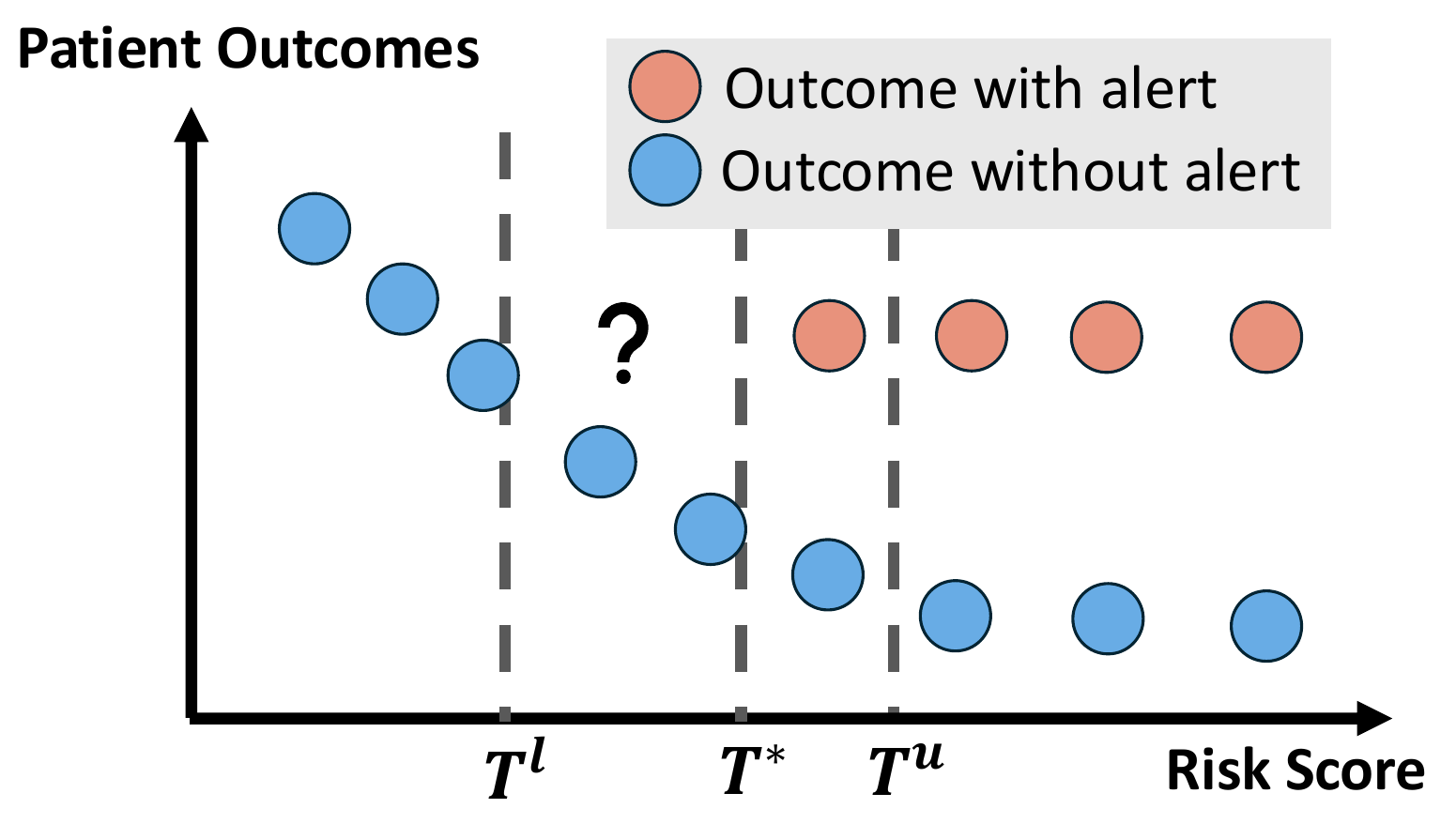}
    }
    \caption{Illustration from~\cref{exp:motivating_example}, demonstrating the challenge of our task. We can use data from both the control arm, that does not raise alerts, and the treatment arm, that raises alerts for patients with risk scores greater than $T^*$, to infer patient outcomes when $\pi_l$ does not raise alerts for patients with risk scores less than $T^l$. Next, we can use data from $\pi_1$ to infer patient outcomes when $\pi_l$ raises alerts for patients with risk scores greater than $T^*$. However, we do not know what patient outcomes are when $\pi_l$ raises alerts for patients with risk scores between $T^l$ and $T^*$. 
    }
    \label{fig:illustration}
\end{figure}

In order to estimate the policy value in~\cref{eq:target_estimand}, we introduce a few key assumptions that relate outcomes under different hypothetical models / policies.  First, since it is unlikely that our target policy $\pi_e$ has exactly the same performance $M$ as policies trialed during the RCT, we need to assume a relationship between outcomes and model performance.
\begin{assumption}[Performance Monotonicity]\label{asmp:more_reliable_better}
    Potential outcomes are non-decreasing in model performance, i.e., if $m_i < m_j$, then for all $a \in \cA$, 
    \begin{equation*}
        Y(A = a, M = m_i) \leq Y(A = a, M = m_j)
    \end{equation*}
\end{assumption}
\Cref{asmp:more_reliable_better} says that improvements in model performance do not harm patient outcomes, given a fixed action. For instance, in the context of~\cref{exp:motivating_example}, we might expect that, for a given patient, having an alarm raised by a high-performance model would not lead to worse outcomes than if that alert had been raised by a model with frequent false alarms. Note that this assumption is stated with a fixed action $A = a$ and does \textit{not} imply that improving performance alone is guaranteed to improve outcomes -- a change in model performance is generally associated with a change in outputs, which may have its own effect on outcomes. In~\cref{sec:simulation}, we give a case where improved overall performance (accuracy) is associated with worse outcomes. \Cref{asmp:more_reliable_better}, however, may not always hold; for instance, clinicians may begin paying less attention to patients that receive a low risk score from the DST, even if it is wrong, as their trust in the system increases. Such a scenario would violate~\cref{asmp:more_reliable_better}. To address this, we propose a method for falsifying~\cref{asmp:more_reliable_better} below.

\begin{restatable}[Falsification of~\cref{asmp:more_reliable_better}]{proposition}{FalsifyMonotonicity}\label{prop:falsify_monotonicity}
    Let $\cX$ denote the full space of possible covariate values. Under~\cref{asmp:scm}, given data from an RCT that includes at least two trialed models $\pi_1$ and $\pi_2$ with different levels of performance $f_M(\pi_1) < f_M(\pi_2)$, and whose actions agree on a non-empty set of patients $\cX_{\text{agree}} \coloneqq \{x \in \cX \mid \pi_1(x) = \pi_2(x) \}$ such that $P(X \in \cX_{\text{agree}}) > 0$, the observation that
    \begin{equation*}
        \E[Y \mid X \in \cX_{\text{agree}}, \Pi = \pi_2] < \E[Y \mid X \in \cX_{\text{agree}}, \Pi = \pi_1],
    \end{equation*}
    implies that~\cref{asmp:more_reliable_better} is false.
\end{restatable}
The proof for \cref{prop:falsify_monotonicity}, along with all other proofs, is given in \cref{app:proofs}. While it is not possible to guarantee that~\cref{asmp:more_reliable_better} is true in general (over all possible models), it has observable implications in an RCT that we can check. In particular, this result suggests a simple hypothesis test that we can use to falsify~\cref{asmp:more_reliable_better}: compare two empirical means in the data and check if outcomes are lower under $\pi_2$ than under $\pi_1$ on those cases where $\pi_1$ and $\pi_2$ agree on their actions. 

In our discussions thus far, we have considered the control arm to be just another policy. However, this framework creates practical difficulties when considering the model performance of a control arm. For instance, suppose the relevant metric for model performance is the false positivity rate (as in~\cref{exp:motivating_example}); then, $M$ is not clearly defined because the control arm never raises alerts. Alternatively, if the relevant metric were model accuracy under no deferral (as in~\cref{exp:motivating_example_2}), then the performance of the control arm would be similarly undefined. One way of resolving this tension is to presume the existence of a \enquote{neutral} action (e.g., not raising an alert, or deferring to clinicians).

\begin{assumption}[Neutral Actions]\label{asmp:A_control_value}
    There exists a \enquote{neutral action} $a_0 \in \cA$ such that the potential outcome of $Y$ under $a_0$ does not depend on model performance $M$.  That is, for any two values $m_1, m_2$, including when $m_1 \neq m_2$,
    \begin{equation}
    Y(A=a_0, M=m_1) = Y(A=a_0, M = m_2),
    \end{equation}
    and in these cases we use the shorthand $Y(A = a_0)$ to reflect the fact that the outcome does not depend on $M$.
\end{assumption}

\Cref{asmp:A_control_value} is a sufficient condition for leveraging data from the control arm of an RCT in our setting, and it implies that,  when a model output is \enquote{neutral} (e.g., no alert in~\cref{exp:motivating_example}, or deferral in~\cref{exp:motivating_example_2}), decision-makers act as they would if no model were deployed. However, note that there may not always exist a ``neutral'' output, especially if decision-makers tend to pay attention to model performance for all possible model outputs\footnote{While all of our results make use of~\cref{asmp:A_control_value}, they can also be re-written to hold if~\cref{asmp:A_control_value} is false by re-defining $a_0$ as some placeholder model output that is never observed under any model (including $\pi_e$), such that indicators like $\1{\pi_e(x) = a_0}$ are always zero.}. Again, we propose a method for falsifying~\cref{asmp:A_control_value} below.

\begin{restatable}[Falsification of~\cref{asmp:A_control_value}]{proposition}{FalsifyNeutralActions}\label{prop:falsify_neutral_actions}
    Under~\cref{asmp:scm}, given data from an RCT that includes at least two trialed models $\pi_1$ and $\pi_2$ with different levels of performance $f_M(\pi_1) < f_M(\pi_2)$, and which both models take the neutral action $a_0$ on a non-empty set of patients $\cX_{a_0} \coloneqq \{x \in \cX \mid \pi_1(x) = \pi_2(x) = a_0 \}$ such that $P(X \in \cX_{a_0}) > 0$, the observation that 
    \begin{equation*}
        \E[Y \mid X \in \cX_{a_0}, \Pi = \pi_2] \neq \E[Y \mid X \in \cX_{a_0}, \Pi = \pi_1],
    \end{equation*}
    implies that~\cref{asmp:A_control_value} is false.
\end{restatable}
Similar to~\cref{prop:falsify_monotonicity}, \cref{prop:falsify_neutral_actions} suggests a simple hypothesis test that can be used to falsify~\cref{asmp:A_control_value}: compare two empirical means in the data to test if outcomes under $\pi_1$ and $\pi_2$ are significantly different on the cases where they both choose $a_0$ as the model output. Of particular interest is the scenario where a control arm exists, and we are interested in checking whether outcomes under the control arm (e.g., not alerting in~\cref{exp:motivating_example}) coincide with outcomes in a treatment arm where the trialed model agrees with the control arm (e.g., does not raise alerts in~\cref{exp:motivating_example}).

\begin{example}[continues=exp:motivating_example]
    Consider an evaluation policy $\pi_u(x) = \1{r(x) > T^u}$ where the threshold for alerting $T^u > T^*$ is greater than the threshold used in the original trial and where the performance of $\pi_u(x)$ (e.g., the precision) is greater than that of the original trialed policy $\pi_1$. \cref{fig:illustration} gives a visual illustration of such a policy. In this scenario, under~\cref{asmp:more_reliable_better,asmp:A_control_value}, we can intuitively infer a lower bound on the policy value of $\pi_u$ using the outcomes of both (a) patients with $r(x) \leq T^u$ who did not receive alerts in the trial (either because they were in the control arm or because $\pi_1$ did not raise alerts), and (b) those patients with $r(x) > T^u$ who did receive alerts under $\pi_1$.
\end{example}

While these assumptions are sufficient in some scenarios, they do not yield meaningful bounds when a new policy takes actions (i.e., a new model produces outputs) on a given case that was never seen for similar cases in the RCT.
\begin{example}[continues=exp:motivating_example]
    Consider the evaluation policy $\pi_l(x) = \1{r(x) > T^l}$ where the threshold for alerting $T^l < T^*$ is \textbf{less} than the threshold used in the original trial. Regardless of the performance of $\pi_l$ in this scenario, even under~\cref{asmp:more_reliable_better,asmp:A_control_value}, we have no way to infer outcomes under $\pi_l$ for the individuals where $r(x) \in [T^l, T^*]$. These correspond to a set of \enquote{never alerted} individuals where $\pi_l$ raises an alert but where neither the control arm nor the trialed policy $\pi_1$ raised an alert.
\end{example}

To resolve this fundamental uncertainty, it is sufficient to know that outcomes $Y$ are bounded, such that we can provide some bounds on expected outcomes in the evaluation of policies that take never-before-seen actions.
\begin{assumption}[Bounded Outcomes]\label{asmp:bounded_Y}
    There exists constants $Y_{\text{min}}, Y_{\text{max}}$ such that $Y_{\text{min}} \leq Y \leq Y_{\text{max}}$.
\end{assumption}

\textbf{Aside: Why require the performance assumption?}
We pause to reflect on the importance of the assumption (implicit in~\cref{asmp:scm}) that implies that our choice of model $\Pi$ impacts outcomes, not only through the outputs $A$, but also through model performance $M$. Broadly speaking, this assumption is not only intuitive from a real-world perspective, but it also has the welcome side-effect of ruling out nonsensical conclusions about trial design. For instance, there are trivial ways to satisfy the condition that, for every value of $\cX$, there exists some model in the trial that matches the output of $\pi_e$.  In the context of~\cref{exp:motivating_example}, for instance, one could trial an alerting system that simply \textit{always raises alerts for every patient}, alongside a control arm that never raises alerts. Then, the requirement that we observe what happens to patients both under no alerts and under alerts for each $x \in \cX$ would be satisfied, eliminating any challenges related to coverage. The assumption that model accuracy $M$ impacts outcomes gives a formal rationale for why this type of trial design is nonsensical: The observed impact of this \enquote{always alert} policy would likely be minimal, or even harmful, compared to never raising alerts due to the negative impact of extremely poor accuracy.

We will shortly present our main result: Under our data-generating process (\cref{asmp:scm}) and the assumptions above (\cref{asmp:more_reliable_better,asmp:A_control_value,asmp:bounded_Y}), we can compute tight bounds on expected outcomes under any proposed model $\pi_e$. First, however, we will define some useful notation for conveying our results, which builds upon the intuition above.

\begin{restatable}[Policy/Model Sets]{definition}{PartitionsDef}\label{def:partitions}
    For each value of $x \in \cX$, we define the sets of trialed policies/models (possibly none) that agree with $\pi_e(x)$ and subsets of this set based on the performance characteristics of those trialed models\footnote{All these sets are defined with respect to the model $\pi_e$ and could be written more precisely with $\pi_e$ as an argument instead of in the superscript (e.g., $\mathbf{\Pi}(x, \pi_e)$) but we use the superscript notation for conciseness.}.
    \begin{align*}
        \Pe(x) &\coloneqq \{\pi \in \varPi \mid \pi(x) = \pi_e(x)\} \\
        \Pel(x) &\coloneqq \{\pi \in \varPi \mid \pi(x) = \pi_e(x), f_M(\pi) \leq f_M(\pi_e)\} \\
        \Peg(x) &\coloneqq \{\pi \in \varPi \mid \pi(x) = \pi_e(x), f_M(\pi) \geq f_M(\pi_e)\}
    \end{align*}
    We also further define subsets of $\Pel$ and $\Peg$ that contain only the next-worst or next-best performing model\footnote{Where relevant, we use the convention that $\argmin_{\pi \in \varnothing}(f_M(\pi)) = \varnothing$.}.
    \begin{align*}
        \Petl(x) &\coloneqq \argmax_{\pi \in \Pel(x)} f_M(\pi), \\
        \Petg(x) &\coloneqq \argmin_{\pi \in \Peg(x)} f_M(\pi)
    \end{align*}
\end{restatable}

\begin{remark}
    To summarize our notation related to deployed and new models, we use $\varPi$ to denote the space of all deployed models in the RCT, $\Pi$ to refer to the variable representing the deployed model, $\pi$ to denote a specific deployed model, $\pi_e$ to denote the new model we are evaluating, and $\Pi$ in boldface with the superscript $e$, such as $\Pe(x), \Pel(x), \Petl(x), \Peg(x),$ and $\Petg(x)$, to denote functions that take a value of $x \in \cX$ as input and return a set of models satisfying various criteria.
\end{remark}

Using~\cref{def:partitions}, we can give precise lower and upper bounds on the performance of a new model $\pi_e$. 

\begin{restatable}{theorem}{BoundsWithControl}\label{thm:bounds_with_control}\label{THM:BOUNDS_WITH_CONTROL}
    Given the data generating process in~\cref{asmp:scm}, and under~\cref{asmp:bounded_Y,asmp:more_reliable_better,asmp:A_control_value}, the policy value of a model / policy $\pi_e$ is bounded as
    \begin{equation}
       L(\pi_e) \leq \E[Y(A=\pi_e, M=f_M(\pi_e))] \leq U(\pi_e),
    \end{equation} 
    where 
    \begin{align}
        L(\pi_e) = \E\big[&\colora{\1{\pi_e \neq a_0}}\big( \nonumber \\
        &\colorb{\1{\Petl(X) \neq \varnothing}} \E[Y \mid X, \Pi \in \Petl(X)] \label{eq:lower1} \\
        +&\colorb{\1{\Petl(X) = \varnothing}} Y_{min} \big) \label{eq:lower2} \\
        +&\colora{\1{\pi_e = a_0}}\big( \nonumber \\
        &\colorb{\1{\Pe(X) \neq \varnothing}} \E[Y \mid X, \Pi \in \Pe(X)] \label{eq:lower3} \\
        +&\colorb{\1{\Pe(X) = \varnothing}} Y_{min}\big)\big] \label{eq:lower4} \\
        U(\pi_e) = \E\big[&\colora{\1{\pi_e \neq a_0}}\big( \nonumber \\
        &\colorb{\1{\Petg(X) \neq \varnothing}} \E[Y \mid X, \Pi \in \Petg(X)] \nonumber \\
        +&\colorb{\1{\Petg(X) = \varnothing}} Y_{max}\big) \nonumber \\
        +& \colora{\1{\pi_e = a_0}}\big( \nonumber \\
        &\colorb{\1{\Pe(X) \neq \varnothing}} \E[Y \mid X, \Pi \in \Pe(X)] \nonumber \\
        +& \colorb{\1{\Pe(X) = \varnothing}} Y_{max}\big)] \nonumber
    \end{align}
    These bounds are still valid if we replace $\Petl(X)$ with $\Pel(X)$ and $\Petg(X)$ with $\Peg(X)$.
\end{restatable}
We give intuition for the construction of the lower bound. First, note that each value of $x \in \cX$ makes a contribution to the construction of the lower bound based on which \colora{green} and which \colorb{purple} indicator it activates. In \cref{eq:lower1}, we consider model outputs that are not the neutral action $a_0$ and values of $X$ at which $\Petl(X)$ is non-empty. That is, there is at least one trialed model agreeing in output with $\pi_e$ that also has less than or equal performance. Here, we use outcome data from trial arms with the next-worst performance to infer outcomes. Next, \cref{eq:lower2} represents values of $X$ where $\pi_e$ does not output the neutral action and there are no agreeing trialed models with less than or equal performance. In this case, we use $Y_{min}$ to lower bound outcomes as we have no trial data on outcomes under such model output. In \cref{eq:lower3}, the new model outputs the neutral action $a_0$, and there is at least one trial model agreeing in output. We can then use data from any trial models that agreed in output to infer outcomes as model performance does not influence outcomes under the neutral action. Finally, \cref{eq:lower4} represents no trial models agreeing in an output of the neutral action; here, we lower bound by $Y_{min}$. The intuition for the upper bound follows similarly.

The lower and upper bounds in \cref{thm:bounds_with_control} can be constructed by iterating over all possible values of $X$, determining whether the model output is a neutral action, checking whether there are agreeing trial models with appropriate levels of performance, and taking the weighted average of the appropriate bounds given the observations above over $X$. In Appendix~\ref{app:algorithmic_bounds}, we give an algorithm for constructing the bounds proposed in~\cref{thm:bounds_with_control} in this manner.

\begin{restatable}[Tightness of bounds in~\cref{thm:bounds_with_control}]{theorem}{NonIdBounds}\label{prop:non-id}
    For any observational distribution $P(X, Y, A, M, \Pi, D)$ consistent with the assumptions of~\cref{thm:bounds_with_control}, there exist two structural causal models $\cM_L, \cM_U$ such that both are consistent with \cref{asmp:scm,asmp:more_reliable_better,asmp:A_control_value,asmp:bounded_Y}, both give rise to that same observational distribution $P$, and where the policy value of any policy $\pi_e$ under $\cM_L, \cM_U$ is given by $L(\pi_e), U(\pi_e)$ from~\cref{thm:bounds_with_control}, respectively. Hence, these bounds cannot be improved without further assumptions.
\end{restatable}
Note that~\cref{prop:non-id} implies that the tightest possible bounds require the use of $\Petl(x)$ and $\Petg(x)$ in place of $\Pel(x)$ and $\Peg(x)$, respectively. This requirement arises because we may get tighter lower-bounds (and similarly, tighter upper-bounds) by using only outcomes under the \enquote{next-worst/best} performing model, rather than averaging over outcomes under all worse/better-performing models. Nonetheless, it may be useful to use $\Pel(x)$ instead of $\Petl(x)$ in some scenarios due to sample-size concerns, especially if outcomes $Y$ do not vary significantly with $M$.

\textbf{When is exact identification possible?} To better understand conditions for agreement of upper and lower bounds, we directly consider the width of these bounds.
\begin{restatable}[Bound Decomposition]{proposition}{BoundDecomposition}\label{prop:bound_decomposition}
    The gap between the bounds in~\cref{thm:bounds_with_control} can be written as 
    \begin{align*}
        U(\pi_e) - L(\pi_e) &= \E[\delta(X, Y, \Pi)]
    \end{align*}
    where 
    \begin{align}
        &\delta(X, Y, \Pi) = \nonumber \\
        &\colorc{\1{\Pe(X) = \varnothing}}(Y_{\text{max}} - Y_{\text{min}}) \label{eq:diff_1} \\
        +&\colora{\1{\pi_e \neq a_0}} \big[ \nonumber \\
        & \colorb{\1{\Petl(X) \neq \varnothing, \Petg(X) \neq \varnothing}} \cdot\label{eq:diff_2} \\ 
        &(\E[Y \mid X, \Pi \in \Petg(X)] - \E[Y \mid X, \Pi \in \Petl(X)]) \nonumber \\ 
        +& \colorb{\1{\Petl(X) \neq \varnothing, \Petg(X) = \varnothing}} \cdot \label{eq:diff_3} \\
        & (Y_{\text{max}} - \E[Y \mid X, \Pi \in \Petl(X)]) \nonumber \\
        +& \colorb{\1{\Petl(X) = \varnothing, \Petg(X) \neq \varnothing}} \cdot \label{eq:diff_4} \\
        & (\E[Y \mid X, \Pi \in \Petg(X)] - Y_{\text{min}}) \big]  \nonumber 
    \end{align}
    Moreover, $\delta(X, Y, \Pi) \geq 0$ almost surely under the assumptions of~\cref{thm:bounds_with_control}.
\end{restatable}

Note that the only way to achieve point identification (a gap of zero in~\cref{prop:bound_decomposition}) is if each component (\cref{eq:diff_1,eq:diff_2,eq:diff_3,eq:diff_4}) is equal to zero. Here, a value of $x \in \cX$ only contributes to the bound decomposition if $x$ satisfies a \colorc{blue} indicator or both a \colora{green} and \colorb{purple} indicator. \cref{eq:diff_1} captures uncertainty that arises in the subset of the population, indexed by $x \in \cX$, for which $\{\Pe(x) = \varnothing\}$ where \textit{no trialed model agrees with the output of the model $\pi_e$} regardless of whether the new model outputs $a_0$. For this term to be zero, there must exist some trialed model that agrees with the action taken by $\pi_e$ for every value of $x \in \cX$.  Second,~\cref{eq:diff_2} reflects differences in the bounds for the subpopulation where the evaluation policy $\pi_e$ takes a non-neutral action ($\pi_e(X) \neq a_0$), and the trialed models that agree with $\pi_e$ have at least one of equal performance\footnote{When there exists an agreeing trial model with equal performance, both $\{\Petl(X) \neq \varnothing\}$ and $\{\Petg(X) \neq \varnothing\}$ are true according to their definitions.} or include both better-performing and worse-performing models. Here, the outcomes under the better/worse-performing models give an upper/lower bound on the outcomes under $\pi_e$. For this term to be zero, either there exists an agreeing trial model with equal performance, this subpopulation $\{\Petl(x) \neq \varnothing, \Petg(x) \neq \varnothing\}$ is empty, or the outcomes under the better and worse-performing models coincide\footnote{Equivalence of conditional outcomes could occur if differences in performance do not impact outcomes for the range of performances tested. Note that~\cref{asmp:more_reliable_better} allows for this occurrence, as it does not assume a strict inequality.}. \Cref{eq:diff_3,eq:diff_4} capture subpopulations where the only models that agree with $\pi_e$ are either worse-performing (giving a lower bound, but no meaningful upper bound) or better-performing (giving an upper bound, but no meaningful lower bound) when the new model does not output $a_0$, and these sets must be empty for these terms to be zero. 

\textbf{How can we estimate these bounds from data?} \Cref{cor:empirical_estimation} below implies a simple estimator that can be used to estimate bounds (and provide asymptotically valid confidence intervals) on the policy value for a new policy $\pi_e$ without the need for training auxiliary models.

\begin{restatable}{proposition}{EmpiricalEstimation}\label{cor:empirical_estimation}
    The bounds in \cref{thm:bounds_with_control} can be written $L(\pi_e) = \E[\psi_L(Y, X, \Pi)]$ and $U(\pi_e) = \E[\psi_U(Y, X, \Pi)]$, where $\psi_L$ and $\psi_U$ are defined as follows
    \begin{align*}
        &\psi_L(Y, X, \Pi) \\
        &\coloneqq \begin{cases}
            Y \cdot \frac{\1{\Pi \in \Petl(X)}}{P(\Pi \in \Petl(X))}, &\ \text{if } \colorb{\Petl(X) \neq \varnothing}, \colora{\pi_e(X) \neq a_0} \\
            Y_{min}, &\ \text{if } \colorb{\Petl(X) = \varnothing}, \colora{\pi_e(X) \neq a_0} \\
            Y \cdot \frac{\1{\Pi \in \Pe(X)}}{P(\Pi \in \Pe(X))}, &\ \text{if } \colorb{\Pe(X) \neq \varnothing}, \colora{\pi_e(X) = a_0} \\
            Y_{min}, &\ \text{if } \colorb{\Pe(X) = \varnothing}, \colora{\pi_e(X) = a_0} \\
        \end{cases} \\
        &\psi_U(Y, X, \Pi) \\
        &\coloneqq \begin{cases}
            Y \cdot \frac{\1{\Pi \in \Petg(X)}}{P(\Pi \in \Petg(X))}, &\ \text{if } \colorb{\Petg(X) \neq \varnothing}, \colora{\pi_e(X) \neq a_0} \\
            Y_{max}, &\ \text{if } \colorb{\Petg(X) = \varnothing}, \colora{\pi_e(X) \neq a_0} \\
            Y \cdot \frac{\1{\Pi \in \Pe(X)}}{P(\Pi \in \Pe(X))}, &\ \text{if } \colorb{\Pe(X) \neq \varnothing}, \colora{\pi_e(X) = a_0} \\
            Y_{max}, &\ \text{if } \colorb{\Pe(X) = \varnothing}, \colora{\pi_e(X) = a_0} \\
        \end{cases}
    \end{align*}
    Moreover, since $\psi_L, \psi_U$ are known functions of the data, these bounds can be estimated as 
    \begin{align*}
    \hat{L}(\pi_e) &\coloneqq n^{-1} \sum_i \psi_L(Y_i, X_i, \Pi_i)\\
    \hat{U}(\pi_e) &\coloneqq n^{-1} \sum_i \psi_U(Y_i, X_i, \Pi_i) 
    \end{align*}
    where $\sqrt{n} (L - \hat{L}) \stackrel{d}{\rightarrow} N(0, \sigma^2(\psi_L))$ where $\sigma^2(\psi_L)$ is the variance of $\psi_L$ and $\stackrel{d}{\rightarrow}$ denotes convergence in distribution, with similar convergence of $\hat{U}$, and hence
    \begin{align*}
        \bigg[&\hat{L}(\pi_e) - \Phi^{-1}\left(1 - \frac{\alpha}{2}\right) \cdot \frac{\hat{\sigma}(\psi_L)}{\sqrt{n}}, \\
        &\hat{U}(\pi_e) + \Phi^{-1}\left(1 - \frac{\alpha}{2}\right)\cdot \frac{\hat{\sigma}(\psi_U)}{ \sqrt{n}}\bigg]
    \end{align*} is an asymptotically valid ($1-\alpha$)-confidence interval, where $\hat{\sigma}(\psi)$ is the empirical standard deviation of $\psi$ and $\Phi^{-1}$ is the inverse of the standard normal CDF.
\end{restatable}

\Cref{cor:empirical_estimation} gives a straightforward way of estimating bounds from data as empirical means over the RCT dataset. To give intuition, terms like $P(\Pi \in \Peg(X))$ are known by design (though they depend on $X$), since $P(\Pi)$ is assumed to be known and the trialed policies $\varPi$ are known, and so asymptotic normality is straightforward to demonstrate\footnote{The statistical efficiency of these bounds could potentially be improved by using a doubly-robust-style estimator that incorporates an estimate of terms, such as $\mu_{=}(X) \coloneqq E[Y \mid X = x, \Pi \in \Petl(X)]$, but we present this simpler estimator for ease of exposition and understanding.}. In \cref{sec:simulation}, we use \cref{cor:empirical_estimation} to estimate the bounds for the effect of deploying a new model.

\section{Recommendations For Pre-trial Design}
\label{sec:implications}

\textbf{Recommendation: Conduct trials with multiple models.} Our results suggest the utility of trialing multiple models in a cluster RCT that vary in their outputs on different patient populations and which exhibit a range of reasonable performance characteristics. First of all, doing so gives more flexibility to estimate the effect of new models if there are sizeable populations where trialed models raise different outputs. Second, falsification of our main assumptions (\cref{asmp:more_reliable_better,asmp:A_control_value}) can be done using data from patient populations where the outputs of different models agree. 

{\bf Recommendation: Use previous trial data to inform deployment of new models.} Given a set of models or policies included in a trial, it may be tempting to conclude that an updated model that is more accurate (on average) should be deployed. However, this conclusion may be flawed given the goal of improving patient outcomes. For instance, a model that is more accurate, but achieves higher accuracy by sacrificing performance on some important subpopulation, may ultimately lead to worse outcomes. Indeed, in~\cref{sec:simulation}, we give a simple simulated example where the optimal model is not the model with the best performance. Given the potential impacts of deploying ML models in practice, it is paramount that we use previous trial data to draw inference on outcomes of interest prior to deploying a new model.

\section{Simulation Study}
\label{sec:simulation}

We now describe a simple simulated example, inspired by~\cref{exp:motivating_example}, that demonstrates the results derived in \cref{sec:identification} and how our proposed method allows for more robust comparisons between models. We consider machine learning models that alert clinicians to the near-term onset of some disease, denoted by $O \in \{0, 1\}$. We simulate a cluster RCT with three arms: A control arm, denoted as a policy $\pi_0$ that never alerts, and two arms where models $\pi_1$ and $\pi_2$ are deployed, respectively. For simplicity, we consider model performance $M$ to be the ground-truth accuracy in predicting disease onset. Our outcome of interest, denoted by $Y \in \{0, 1\}$, is patient survival, and $X$ represents a baseline health characteristic. For the sake of an interpretable simulation, $X$ takes on values uniformly in $\{0, 1, 2, 3\}$, and $O$ and $Y$ are Bernoulli random variables. The probability of $O=1$ depends only on $X$ whereas the probability of $Y=1$ depends on $X$, $A$, and $M$. The trialed models are defined as $\pi_1 = \1{X = 1}$ and $\pi_2 = \1{X \in \{2, 3\}}$. 

\begin{figure}[t]
    \includegraphics[scale=0.46]{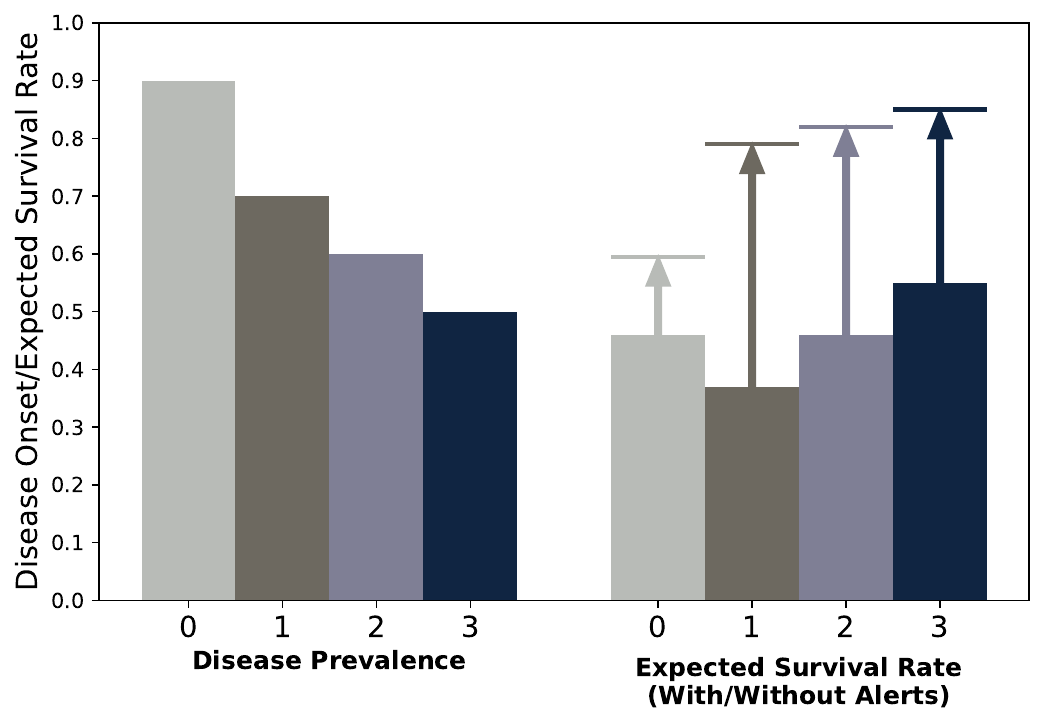}
    \caption{Bar graph giving a visual representation of the data generating process in the simulation study. The bars on the left depict the likelihood of disease onset for varying levels of baseline health $X$, and the bars on the right depict the likelihood of survival, assuming model performance is fixed at $m=0.5$. The arrows above the bars on the right show how survival rates change for each level of $X$ when patients receive an alert from an alerting model.}
    \label{fig:dgp_plot}
\end{figure}
\begin{table}[t]
    \centering
    \caption{Probability of developing disease ($\E[O]$) and expected survival rate ($\E[Y(a,m)]$) for the simulation study.}
    \label{tab:dgp_O_Y}
    \scalebox{0.75}{
    \begin{tabular}{ccc}
         & \multicolumn{1}{|c}{$X=0$} & \multicolumn{1}{|c}{$X=1$} \\
         \hline
         $\E[O]$ & \multicolumn{1}{|c}{$0.9$} & \multicolumn{1}{|c}{$0.7$} \\
         $\E[Y(a, m)]$ & \multicolumn{1}{|c}{$0.46+((1 + m)/2) \cdot 0.18 a$} & \multicolumn{1}{|c}{$0.37+((1 + m )/2) \cdot 0.56 a$} \\
         \\
         & \multicolumn{1}{|c}{$X=2$} & \multicolumn{1}{|c}{$X=3$} \\
         \hline
         $\E[O]$ & \multicolumn{1}{|c}{$0.6$} & \multicolumn{1}{|c}{$0.5$} \\
         $\E[Y(a, m)]$ & \multicolumn{1}{|c}{$0.46+((1 + m)/2) \cdot 0.48 a$} & \multicolumn{1}{|c}{$0.55+((1 + m )/2) \cdot 0.4 a$} \\
    \end{tabular}
    }
\end{table}

\Cref{tab:dgp_O_Y} gives expected values of $O$ and $Y(a, m)$ over all possible values of $X$, satisfying \cref{asmp:scm,asmp:A_control_value,asmp:bounded_Y}; \cref{fig:dgp_plot} visualizes the same information for fixed model performance. The aggregate patterns reflect the following process: First, we assume that all patients who do not develop disease will survive. Then, among the sickest individuals ($X = 0$), the survival rate among patients with disease is 40\% without alerts\footnote{To align this explanation with \cref{tab:dgp_O_Y}, note that all of the 10\% without disease and $40\%$ of the $90\%$ with disease survive, yielding the overall survival rate of 46\% without alerts.} and between 50-60\% with alerts, as influenced by $m$, where 50\% corresponds to a model with zero accuracy, and 60\% is achieved by a model with perfect accuracy. These patterns reflect the intuition that survival is improved (when alerts are raised) under a model perceived to be more accurate. However, the survival rate (among those with disease) is much improved for other groups, going from 10\% without alerts to 50-90\% with alerts\footnote{For explicit computations of survival rates with and without alerts, see~\cref{app:dgp}.}. Since the actual prevalence of disease varies across these groups, and since alerts only help those with disease, the overall causal effect is largest among those where $X = 1$. \Cref{app:dgp} describes and explains this data generating process in detail.
\begin{figure}[t]
    \includegraphics[scale=0.46]{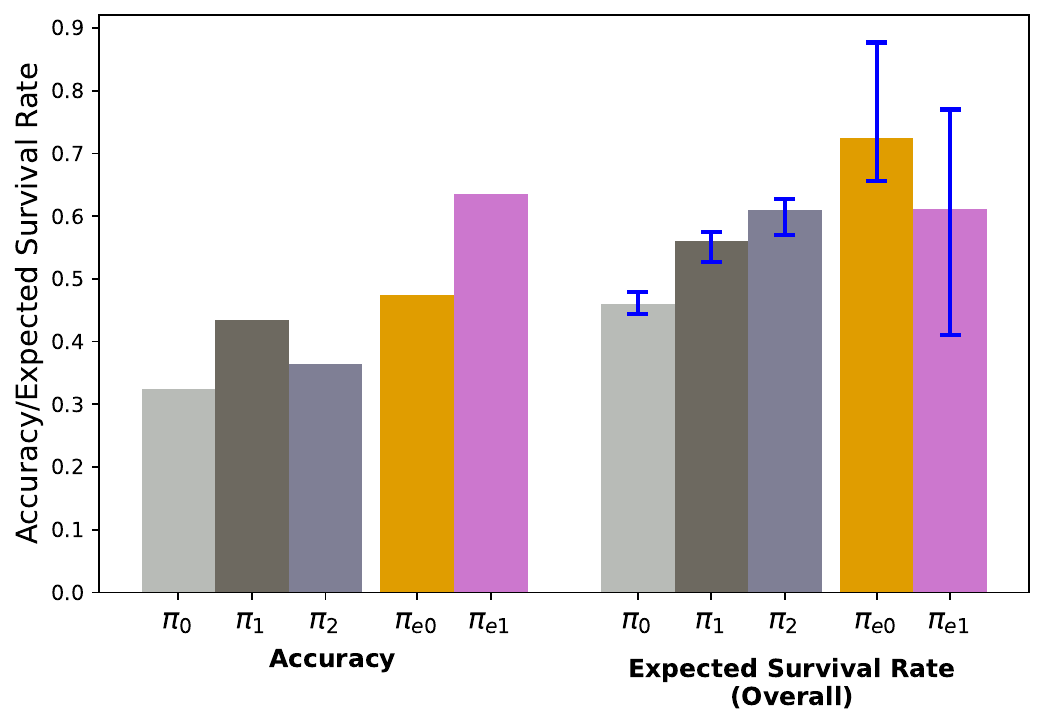}
    \caption{Bar graph showing accuracy at predicting disease and expected patient survival rates of different simulated alerting models. Bars denote ground-truth values, and blue intervals denote estimated bounds using our approach.  Because $\pi_{e0}$ and $\pi_{e1}$ were not trialed in the RCT, their bounds are wider. Despite $\pi_{e1}$ having the greatest raw accuracy, $\pi_{e0}$ has the most positive effect on patient survival rates.}
    \label{fig:sim}
\end{figure}

We consider the effect of deploying two new models, $\pi_{e0} = \1{X \in \{1, 2, 3\}}$ and $\pi_{e1} = \1{X \in \{0, 1\}}$, whose accuracies at predicting onset are computable using~\cref{tab:dgp_O_Y}. We then estimate bounds on $\E[Y(A=\pi, M=f_M(\pi))]$ for each model (including confidence intervals to incorporate finite-sample uncertainty, as described under~\cref{cor:empirical_estimation}) using data from a simulation with $n=5,000$. \cref{fig:sim} shows the simulation results and illustrates that \textbf{the most accurate model is not always the best}: $\pi_{e1}$ has the greatest accuracy in predicting the onset of disease, but $\pi_{e0}$ has the largest causal impact on patient outcomes.  This reversal occurs since $\pi_{e0}$ raises alerts for patients who stand to benefit the most, whereas $\pi_{e1}$ tends to alert for patients who have little to gain from them.  Moreover, our bounds reflect greater confidence in the (positive) impact of $\pi_{e0}$, since the lower bound for patient outcomes under $\pi_{e0}$ is greater than patient outcomes under all trialed models. Python code implementing this simulation study, implementing the estimation procedure proposed in \cref{cor:empirical_estimation}, and for generating \cref{fig:dgp_plot,fig:sim} are publicly available online at: \url{https://github.com/jacobmchen/just_trial_once}.

\section{Conclusion and Limitations}
\label{sec:conclusion}

In this paper, we discussed methods for estimating the causal impact of new or updated ML and AI models not previously trialed in an RCT. Under the important considerations that ML predictions are deterministic and that clinician trust in ML models play a role in determining their impacts on patient outcomes, we demonstrated how one could estimate lower and upper bounds on the effect of deploying a new model. We further proved tightness of our proposed bounds (\cref{prop:non-id}) and gave inverse probability weighted-style estimators for them (\cref{cor:empirical_estimation}). Given the possibility that key assumptions for employing our methods may not be fulfilled, we also proposed simple strategies for testing and falsifying them. Finally, we concluded with a simulation study to illustrate the application of our method and to highlight its benefits when selecting among model updates. 

However, our work is not without limitations: First, our derived bounds are naturally pessimistic, and, while they cannot be tightened without further assumptions (\cref{prop:non-id}), some additional assumptions may be warranted in certain cases. For instance, we implicitly assume that \enquote{anything can happen} when a model raises an alert on patients who never received alerts in the past.  One could instead assume that alerts are not harmful (except for their impact on performance), or that they are not harmful for a specific patient if the alert is correct. Second, we assume that model performance can be summarized in a single real number, but a more complex representation (e.g., involving subgroup-specific performance, performance adaptation over time, and clinician experience) may be warranted in some applications. Finally, a core limitation of our approach is that we still require an RCT. Using RCT data comes with many benefits:  It allows for greater confidence that core assumptions (e.g., randomization of policies) hold by design, and even allows for checking (in some cases) the core assumptions we make in this paper, as we have shown (\cref{prop:falsify_monotonicity,prop:falsify_neutral_actions}).  However, observational studies (e.g., pre-post studies of model deployments) are often easier to run in practice.  Our hope is that this work can serve as a springboard towards increasing the utility, and thus adoption, of RCTs for ML models deployed in high-risk settings.


\newpage

\begin{acknowledgements}
    The authors are grateful for the helpful feedback of anonymous reviewers as well as fruitful discussions with Jonathan Zhang, Erik Skalnes, and Trung Phung.
\end{acknowledgements}

\bibliography{uai2025-template}

\newpage

\onecolumn

\appendix
\title{Supplementary Material}
\maketitle

\section{Brief Causal Inference Overview}
\label[secinapp]{app:overview}

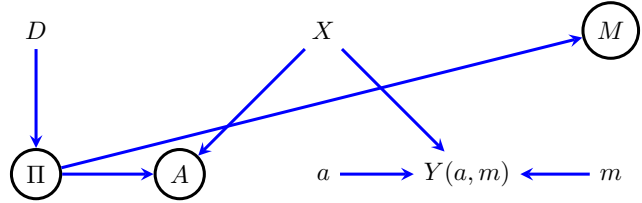
\begin{figure}[ht]
    \centering
    \scalebox{0.95}{
        \begin{tikzpicture}[>=stealth, node distance=2cm]
            \tikzstyle{square} = [draw, thick, minimum size=1.0mm, inner sep=3pt]
            \begin{scope}
                \path[->, very thick]
                node[circle, draw] (a) {$A$}
                node[right of=a] (af) {$a$}
                node[right of=af] (y) {$Y(a,m)$}
                node[above of=a,xshift=2cm] (x) {$X$}
                node[right of=m] (mf) {$m$}
                node[above of=mf, circle, draw] (m) {$M$}
                node[left of=a, circle, draw] (pi) {$\Pi$}
                node[above of=pi] (d) {$D$}
                
                (af) edge[blue] (y)
                (x) edge[blue] (a)
                (x) edge[blue] (y)
                (d) edge[blue] (pi)
                (pi) edge[blue] (a)
                (mf) edge[blue] (y)
                (pi) edge[blue] (m)
                ;
            \end{scope}
        \end{tikzpicture}
    }
    \caption{SWIG showing $\mathcal{G}(a,m)$.}
    \label{fig:swig}
\end{figure}

In this section, we give a brief overview of additional topics in causal inference that we use in our proofs and provide references for further reading on these topics.

First, independencies implied in the full data distribution $P({\bf V})$ can be read off from $\mathcal{G}$ using the d-separation criterion \citep{pearl2014probabilistic}. Next, we follow \citep{richardson2013single} to convert a DAG $\mathcal{G}$, given an intervened on variable $W$, to a single-world intervention graph (SWIG) using a node-splitting transformation as follows: (1) replace all children of $W$ (here, $V$) with the potential outcome random variable $V(w)$, (2) add the intervened on variable $w$ into $\mathcal{G}$ as a new vertex, and (3) change all outgoing edges from $W$ to originate from $w$ instead. This new SWIG is denoted by $\mathcal{G}(w)$. It is possible to intervene on multiple variables in $\mathcal{G}$ by repeatedly applying the node-splitting transformation described above. \cref{fig:swig} shows the SWIG $\mathcal{G}(a, m)$, the SWIG containing the potential outcome random variable $Y(a,m)$ corresponding to the target estimand $\E[Y(A=\pi_e, M=f_M(\pi_e))]$. 

In order to connect potential outcome distributions to observed distributions, \citep{malinsky2019potential} propose a set of three rules known as the {\it potential outcome calculus} (po-calculus). Most salient to our proofs is Rule 2 of po-calculus, which states that $P(V(w) \mid {\bf X}) = P(V \mid {\bf X}, W=w)$ if $V(w) \ci W \mid {\bf X}$ in $\mathcal{G}(w)$, where ${\bf X}$ is any set of random variables in $\mathcal{G}(w)$, including the empty set. Rule 2 is also referred to as the {\it conditional ignorability} assumption commonly used in causal inference, which states that a potential outcome of interest is independent of a treatment of interest conditional on a sufficiently rich set of covariates ${\bf X}$. The benefit of using Rule 2 of po-calculus is that it allows us to use d-separation in a SWIG to directly verify whether conditional ignorability holds given a graph. For instance, note that $Y(a,m) \ci A,M \mid X$ in the SWIG shown in \cref{fig:swig}. Rule 2 of po-calculus thus allows us to conclude that $P(Y(a,m) \mid X) = P(Y \mid A=a, M=m, X)$.

We now formally define some concepts from our discussion above that we use in our proofs.

\begin{corollary}[Consistency]\label{consistency}
    Under~\cref{asmp:scm}, if $A = a, M = m$, then $Y = Y(A = a, M = m)$.
\end{corollary}

\begin{corollary}[Conditional Ignorability]\label{rule_2}
    Under~\cref{asmp:scm}, since $Y(a, m) \ci A, M \mid X$, then $P(Y(a,m) \mid X) = P(Y \mid A=a, M=m, X)$.
\end{corollary}

\section{Algorithmic Viewpoint of Computing Bounds}
\label[secinapp]{app:algorithmic_bounds}

Here, we give an algorithmic construction of the lower and upper bounds presented in~\cref{thm:bounds_with_control}. First, we define sub-algorithms in~\cref{alg:sub_algs} that compute the lower and upper bounds conditional on a value of $x \in \cX$. Next, we define the algorithm that constructs the complete bounds in~\cref{alg:bound_alg} by iterating through all possible values of $x \in \cX$ and computing the weighted average of the bounds for each $x$.

\begin{algorithm}[ht]
    \caption{Sub-Algorithms for computing lower and upper bounds conditional on $X$.}
    \begin{algorithmic}[1]
        \State Input: $x$
        \State Output: Tight lower bound on $Y$ conditional on $X=x$
        \Function{Conditional\_Lower}{$x$}
            \If{$\pi_e(x) \neq a_0$}
                \If{$\Petl(x) \neq \varnothing$}
                    \Return $\E[Y \mid X=x, \Pi \in \Petl(x)]$
                \ElsIf{$\Petl(x) = \varnothing$}
                    \Return $Y_{min}$
                \EndIf
            \ElsIf{$\pi_e(x) = a_0$}
                \If{$\Pe(x) \neq \varnothing$}
                    \Return $\E[Y \mid X=x, \Pi \in \Pe(x)]$
                \ElsIf{$\Pe(x) = \varnothing$}
                    \Return $Y_{min}$
                \EndIf
            \EndIf
        \EndFunction
        \\
        \State Input: $x$
        \State Output: Tight upper bound on $Y$ conditional on $X=x$
        \Function{Conditional\_Upper}{$x$}
            \If{$\pi_e(x) \neq a_0$}
                \If{$\Petg(x) \neq \varnothing$}
                    \Return $\E[Y \mid X=x, \Pi \in \Petg(x)]$
                \ElsIf{$\Petg(x) = \varnothing$}
                    \Return $Y_{max}$
                \EndIf
            \ElsIf{$\pi_e(x) = a_0$}
                \If{$\Pe(x) \neq \varnothing$}
                    \Return $\E[Y \mid X=x, \Pi \in \Pe(x)]$
                \ElsIf{$\Pe(x) = \varnothing$}
                    \Return $Y_{max}$
                \EndIf
            \EndIf
        \EndFunction
    \end{algorithmic}
    \label{alg:sub_algs}
\end{algorithm}
\begin{algorithm}[ht]
    \caption{Algorithm for computing lower and upper bounds for $\E[Y(A=\pi_e, M=f_M(\pi_e)]$.}
    \begin{algorithmic}[1]
        \State $lower\_bound \gets 0$
        \State $upper\_bound \gets 0$
        \For{$x \in \cX$}
            \State $lower\_bound \gets lower\_bound + \textsc{Conditional\_Lower}(x) \times P(X=x)$
            \State $upper\_bound \gets upper\_bound + \textsc{Conditional\_Upper}(x) \times P(X=x)$
        \EndFor \\
        \Return $(lower\_bound, upper\_bound)$
    \end{algorithmic}
    \label{alg:bound_alg}
\end{algorithm}

\section{Data Generating Process of Simulation Study}
\label[secinapp]{app:dgp}

Here, we define in detail the data generating process for the simulation study in~\cref{sec:simulation}. The trialed models are defined as $\pi_0 = 0$, $\pi_1 = \1{X=1}$, and $\pi_2 = \1{X \in \{2, 3\}}$. We use $\text{acc}(\pi_i)$ to denote the accuracy of policy $\pi_i$ at predicting disease $O$. One can verify that this data generating process obeys the structural causal model given in~\cref{asmp:scm}.
\begin{align*}
   X &\sim \text{discrete uniform}[0, 3] \\
   D &\sim \text{discrete uniform}[0, 2] \\
   \Pi &= \pi_D \\
   A &= \pi_D(X) \\
   O &\sim \text{Bernoulli}(\1{X=0}0.9 + \1{X=1}0.7 + \1{X=2}0.6 + \1{X=3}0.5) \\
   M &= \text{acc}(\Pi) \\
   Y &\sim \text{Bernoulli}(\1{X=0}(0.46 + ((1+m)/2) \cdot 0.18A) + \1{X=1}(0.37 + ((1+m)/2) \cdot 0.56A) \\
   &\quad \quad + \1{X=2}(0.46 + ((1+m)/2) \cdot 0.48A) + \1{X=3}(0.55 + ((1+m)/2) \cdot 0.4A)
\end{align*}
We now explain in detail how the probabilities of the survival rate $Y$ are computed. First, conditional on patients who have the disease ($O=1$), the survival rates without and with alerts are as follows:
\begin{table}[ht]
    \centering
    \caption{Probability of survival ($\E[Y(a,m)]$) under no alerts ($\E[Y(0,m)]$) and alerts ($\E[Y(1,m)]$) for the simulation study.}
    \label{tab:survival_rates}
    \scalebox{1}{
    \begin{tabular}{ccc}
         & \multicolumn{1}{|c}{$X=0$} & \multicolumn{1}{|c}{$X=1$} \\
         \hline
         $\E[Y(0, m) \mid O=1]$ & \multicolumn{1}{|c}{$0.4$} & \multicolumn{1}{|c}{$0.1$} \\
         $\E[Y(1, m) \mid O=1]$ & \multicolumn{1}{|c}{$0.4 + ((1+m)/2) \cdot 0.2$} & \multicolumn{1}{|c}{$0.1 + ((1+m)/2) \cdot 0.8$} \\
         \\
         & \multicolumn{1}{|c}{$X=2$} & \multicolumn{1}{|c}{$X=3$} \\
         \hline
         $\E[Y(0, m) \mid O=1]$ & \multicolumn{1}{|c}{$0.1$} & \multicolumn{1}{|c}{$0.1$} \\
         $\E[Y(1, m) \mid O=1]$ & \multicolumn{1}{|c}{$0.1 + ((1+m)/2) \cdot 0.8$} & \multicolumn{1}{|c}{$0.1 + ((1+m)/2) \cdot 0.8$}
    \end{tabular}
    }
\end{table}

When patients do not have the disease ($O=0$), their probability of survival is $1$, that is $\E[Y(a, m) \mid O=0]=1$ for all values of $A$ and $M$. The assumption that lower model performance has a weakening effect on the positive effect of alerting is implicit in~\cref{tab:survival_rates}. For patients with $X=0$ and $O=1$, the survival probability with alerts range from $0.5$ when accuracy is $0$ to $0.6$ when accuracy is $1$. For all other patients, the survival probability with alerts range from $0.5$ when accuracy is $0$ to $0.9$ when accuracy is $1$.

Now, to compute the probabilities of survival given in~\cref{tab:dgp_O_Y}, we apply the law of total probability -- $\E[Y(a, m)] = \E[Y(a, m) \mid O=0]p(O=0) + \E[Y(a, m) \mid O=1]p(O=1)$ -- for each value of $X$. We give this explicit computation in the following table.
\begin{table}[ht]
    \centering
    \caption{Probability of survival ($\E[Y(a,m)]$) computed explicitly when summed over the likelihood of developing the disease ($O=1$).}
    \label{tab:survival_prob}
    \scalebox{1}{
    \begin{tabular}{ccc}
         & \multicolumn{1}{|c}{$X=0$} & \multicolumn{1}{|c}{$X=1$} \\
         \hline
         $\E[Y(a, m)]$ & \multicolumn{1}{|c}{$1 \cdot 0.1 + (0.4 + ((1+m)/2)0.2a) \cdot 0.9$} & \multicolumn{1}{|c}{$1 \cdot 0.3 + (0.1 + ((1+m)/2)0.8a) \cdot 0.7$} \\
         \\
         & \multicolumn{1}{|c}{$X=2$} & \multicolumn{1}{|c}{$X=3$} \\
         \hline
         $\E[Y(a, m)]$ & \multicolumn{1}{|c}{$1 \cdot 0.4 + (0.1 + ((1+m)/2)0.8a) \cdot 0.6$} & \multicolumn{1}{|c}{$1 \cdot 0.5 + (0.1 + ((1+m)/2)0.8a) \cdot 0.5$}
    \end{tabular}
    }
\end{table}

After simplifying the expressions above, we get the probabilities of survival shown in~\cref{tab:dgp_O_Y}. Note that the dependence of $Y$ on $O$ does not violate~\cref{asmp:scm} and does not change~\cref{fig:dag_setup} because $O$ is an unobserved variable on the directed path from $X$ to $Y$.

\section{Proofs}
\label[secinapp]{app:proofs}

\FalsifyMonotonicity*

\begin{proof}
    We will show that if~\cref{asmp:more_reliable_better} holds, then the stated observation will yield a contradiction. We will first show that the~\cref{asmp:more_reliable_better} implies a point-wise inequality over $x$, and then show that this inequality holds when aggregating over $X \in \cX_{\text{agree}}$. First, choose any $x \in \cX_{\text{agree}}$. Then we can write
    \begin{align}
        \E[Y \mid X = x, \Pi = \pi_i] &= \E[Y \mid X = x, A = \pi_i(x), M = f_M(\pi_i), \Pi = \pi_i] \label{eq:proof_monotonicity_1} \\
                                      &= \E[Y(A = \pi_i(x), M = f_M(\pi_i)) \mid X = x, A = \pi_i(x), M = f_M(\pi_i), \Pi = \pi_i] \label{eq:proof_monotonicity_2} \\
                                      &= \E[Y(A = \pi_i(x), M = f_M(\pi_i)) \mid X = x] \label{eq:proof_monotonicity_3}
    \end{align}
    where~\cref{eq:proof_monotonicity_1} follows from the implication $\{\Pi = \pi_i, X = x\} \implies \{A = \pi_i(x), M = f_M(\pi_i)\}$,~\cref{eq:proof_monotonicity_2} follows from consistency (\cref{consistency}), and~\cref{eq:proof_monotonicity_3} follows from conditional ignorability (\cref{rule_2}). Since~\cref{eq:proof_monotonicity_3} holds for both policies $\pi_1, \pi_2$, we can then write that 
    \begin{align}
        &\E[Y \mid X = x, \Pi = \pi_2] - \E[Y \mid X = x, \Pi = \pi_1] \nonumber \\
        &\qquad = \E[Y(A = \pi_2(x), M = f_M(\pi_2)) - Y(A = \pi_1(x), M = f_M(\pi_1)) \mid X = x] \label{eq:proof_monotonicity_4} \\
                                                                      & \qquad \geq 0 \label{eq:proof_monotonicity_5}
    \end{align}
    where~\cref{eq:proof_monotonicity_4} follows from linearity of expectation, and~\cref{eq:proof_monotonicity_5} follows from~\cref{asmp:more_reliable_better}, since $\pi_1(x) = \pi_2(x)$ by construction, and $f_M(\pi_2) > f_M(\pi_1)$. To aggregate, we first observe that $X \ci \Pi$ in our data generating process (as the policies are assigned randomly and independently of $X$).  Hence $P(X \mid \Pi = \pi_2, X \in \cX_{\text{agree}}) = P(X \mid \Pi = \pi_1, X \in \cX_{\text{agree}}) = P(X \mid X \in \cX_{\text{agree}})$.  Accordingly, we can write that 
    \begin{align*}
        &\E[Y \mid X \in \cX_{\text{agree}}, \Pi = \pi_2] - \E[Y \mid X \in \cX_{\text{agree}}, \Pi = \pi_1] \\
        &\qquad = \int_{x} \E[Y \mid X = x, \Pi = \pi_2] dP(x \mid \Pi = \pi_2, X \in \cX_{\text{agree}}) - \int_{x} \E[Y \mid X = x, \Pi = \pi_1] dP(x \mid \Pi = \pi_1, X \in \cX_{\text{agree}}) \\
        &\qquad = \int_{x} \E[Y \mid X = x, \Pi = \pi_2] dP(x \mid X \in \cX_{\text{agree}}) - \int_{x} \E[Y \mid X = x, \Pi = \pi_1] dP(x \mid X \in \cX_{\text{agree}}) \\
        &\qquad = \int_{x} \left(\E[Y \mid X = x, \Pi = \pi_2] - \E[Y \mid X = x, \Pi = \pi_1]\right) dP(x \mid X \in \cX_{\text{agree}}) \\
        &\qquad \geq 0
    \end{align*}
    where the final inequality follows from the point-wise inequality in~\cref{eq:proof_monotonicity_5}, and which directly gives the implication
    \begin{equation*}
        \E[Y \mid X \in \cX_{\text{agree}}, \Pi = \pi_2] \geq \E[Y \mid X \in \cX_{\text{agree}}, \Pi = \pi_1],
    \end{equation*}
    which is contradicted in the case where the stated observation (the relationship $<$ instead of $\geq$) holds.
\end{proof}

\FalsifyNeutralActions*
\begin{proof}
Our proof follows a similar structure to that of~\cref{prop:falsify_monotonicity}. We will show that if~\cref{asmp:A_control_value} holds, then the stated observation would yield a contradiction. We will first show that the~\cref{asmp:A_control_value} implies a point-wise equality over $x$, and then show that this inequality holds when aggregating over $X \in \cX_{a_0}$. First, choose any $x \in \cX_{a_0}$. Then we can write
    \begin{align}
        \E[Y \mid X = x, \Pi = \pi_i] &= \E[Y(A = \pi_i(x), M = f_M(\pi_i)) \mid X = x] \label{eq:proof_neutral_3}
    \end{align}
    using the same argument as in~\cref{prop:falsify_monotonicity} (namely, the implication that $\{\Pi = \pi_i, X = x\} \implies \{A = \pi_i(x), M = f_M(\pi_i)\}$, consistency (\cref{consistency}), and conditional ignorability (\cref{rule_2}). Since~\cref{eq:proof_neutral_3} holds for both policies $\pi_1, \pi_2$, we can then write that 
    \begin{align}
        &\E[Y \mid X = x, \Pi = \pi_2] - \E[Y \mid X = x, \Pi = \pi_1] \nonumber \\
        &\qquad = \E[Y(A = \pi_2(x), M = f_M(\pi_2)) - Y(A = \pi_1(x), M = f_M(\pi_1)) \mid X = x] \label{eq:proof_neutral_4} \\
        &\qquad = \E[Y(A = a_0, M = f_M(\pi_2)) - Y(A = a_0, M = f_M(\pi_1)) \mid X = x] \label{eq:proof_neutral_5} \\
        & \qquad =0 \label{eq:proof_neutral_6}
    \end{align}
    where~\cref{eq:proof_neutral_4} follows from linearity of expectation,~\cref{eq:proof_neutral_5} follows from the fact that $x \in \cX_{a_0}$, and~\cref{eq:proof_neutral_6} follows from~\cref{asmp:A_control_value}, which states that $Y(A=a_0, M = m) = Y(A = a_0, M = m')$ for any $m, m'$.

    To aggregate, we first observe (similar to the proof of~\cref{prop:falsify_monotonicity}) that $X \ci \Pi$ in our data generating process (as the policies are assigned randomly and independently of $X$).  As a result, we can write that 
    \begin{align*}
        &\E[Y \mid X \in \cX_{a_0}, \Pi = \pi_2] - \E[Y \mid X \in \cX_{a_0}, \Pi = \pi_1] \\
        &\qquad = \int_{x} \E[Y \mid X = x, \Pi = \pi_2] dP(x \mid \Pi = \pi_2, X \in \cX_{a_0}) - \int_{x} \E[Y \mid X = x, \Pi = \pi_1] dP(x \mid \Pi = \pi_1, X \in \cX_{a_0}) \\
        &\qquad = \int_{x} \E[Y \mid X = x, \Pi = \pi_2] dP(x \mid X \in \cX_{a_0}) - \int_{x} \E[Y \mid X = x, \Pi = \pi_1] dP(x \mid X \in \cX_{a_0}) \\
        &\qquad = \int_{x} \left(\E[Y \mid X = x, \Pi = \pi_2] - \E[Y \mid X = x, \Pi = \pi_1]\right) dP(x \mid X \in \cX_{a_0}) \\
        &\qquad = 0
    \end{align*}
    where the final equality follows from the point-wise equality in~\cref{eq:proof_neutral_6}, and which directly gives the implication
    \begin{equation*}
        \E[Y \mid X \in \cX_{a_0}, \Pi = \pi_2] = \E[Y \mid X \in \cX_{a_0}, \Pi = \pi_1],
    \end{equation*}
    which is contradicted in the case where the stated observation (an inequality instead of equality) holds.
\end{proof}

\subsection{Proof of Theorem \ref{thm:bounds_with_control}}

Before we give the proof of~\cref{thm:bounds_with_control}, we restate~\cref{def:partitions} from the main text, for ease of reference when reviewing the proof.
\PartitionsDef*

Armed with~\cref{def:partitions}, we first state some useful inequalities that will form the core of the proof for~\cref{thm:bounds_with_control}

\begin{restatable}[Independence under neutral actions]{lemma}{IndepNeutralLemma}\label{lemma:independence_under_neutral}
    Under the assumed data generating process (\cref{asmp:scm}) and \cref{asmp:A_control_value},
    \begin{equation}
        Y \ci M \mid X, A = a_0
    \end{equation}
\end{restatable}
\begin{proof}
    This claim follows from~\cref{asmp:A_control_value} in a straightforward fashion.  Let $m$ and $m'$ be any two distinct values of $M$, then
    \begin{align*}
        P(Y \mid X = x, M = m, A = a_0) &= P(Y(A = a_0, M = m) \mid X = x, M = m, A = a_0) && \text{Consistency}\\
        &= P(Y(A = a_0, M = m') \mid X = x, M = m, A = a_0) && \text{By~\cref{asmp:A_control_value}} \\
        &= P(Y(A = a_0, M = m') \mid X = x, M = m', A = a_0) && Y(a, m) \ci M, A \mid X \\
        &= P(Y \mid X = x, M = m', A = a_0) && \text{Consistency} 
    \end{align*}
    The claim follows from the fact that we have shown equality $P(Y \mid X = x, M = m, A = a_0) = P(Y \mid X = x, M = m', A = a_0)$ for arbitrary $m, m'$.
\end{proof}

\begin{restatable}[Outcome Equalities / Inequalities]{lemma}{BoundLemma}\label{lemma:bound_lemma}
Under \cref{asmp:more_reliable_better,asmp:A_control_value}, the following inequalities hold, where we use $\pi_e$ as shorthand for $\pi_e(x)$,
\begin{align}
    \1{\Petl(x) \neq \varnothing} \E[Y \mid X = x, A=\pi_e, M=f_M(\pi_e)] &\geq \1{\Petl(x) \neq \varnothing} \E[Y \mid X = x, \Pi \in \Petl(x)] \label{eq:inequality_bounds_outcome_less} \\
    \1{\Petg(x) \neq \varnothing} \E[Y \mid X = x, A=\pi_e, M=f_M(\pi_e)] &\leq \1{\Petg(x) \neq \varnothing} \E[Y \mid X = x, \Pi \in \Petg(x)] \label{eq:inequality_bounds_outcome_great} \\
    \1{\Pe(x) \neq \varnothing, \pi_e=a_0} \E[Y \mid X = x, A=\pi_e, M=f_M(\pi_e)] &= \nonumber \\
    \1{\Pe(x) \neq \varnothing, \pi_e=&a_0} \E[Y \mid X = x, \Pi \in \Pe(x)] \label{eq:inequality_bounds_outcome_neutral}
\end{align}
\end{restatable}
\begin{proof}
    The proof for each of these relations follows a similar structure. For each, we only need to consider the relation when the stated indicator (identical on either side) is equal to $1$, since all are trivially true when the indicator is equal to zero.

    For~\cref{eq:inequality_bounds_outcome_less}, the event $\{\Pi \in \Petl(x)\}$ implies $\{A = \pi_e(x), M \leq f_M(\pi_e)\}$ from the definition of $\Petl(x)$ (see~\cref{def:partitions}). By the independence $Y \ci \Pi \mid X, A, M$, we have it that $\E[Y \mid X = x, \Pi \in \Petl(x)] = E[Y \mid X = x, A = \pi_e, M \leq f_M(\pi_e)]$, and this conditional expectation is well-defined when $\Petl(x) \neq \varnothing$. Finally, we make use of~\cref{asmp:more_reliable_better}, which implies that $\E[Y \mid X = x, A=\pi_e, M=f_M(\pi_e)] \geq \E[Y \mid X = x, A=\pi_e, M \leq f_M(\pi_e)]$. The stated inequality follows.

    For~\cref{eq:inequality_bounds_outcome_great}, the event $\{\Pi \in \Petg(x)\}$ implies $\{A = \pi_e(x), M \geq f_M(\pi_e)\}$ from the definition of $\Petg(x)$ (see~\cref{def:partitions}). By the independence $Y \ci \Pi \mid X, A, M$, we have it that $\E[Y \mid X = x, \Pi \in \Petg(x)] = E[Y \mid X = x, A = \pi_e, M \geq f_M(\pi_e)]$, and this conditional expectation is well-defined when $\Petg(x) \neq \varnothing$. Finally, we make use of~\cref{asmp:more_reliable_better}, which implies that $\E[Y \mid X = x, A=\pi_e, M=f_M(\pi_e)] \leq \E[Y \mid X = x, A=\pi_e, M \geq f_M(\pi_e)]$. The stated inequality follows.

    Finally, for~\cref{eq:inequality_bounds_outcome_neutral}, we can observe that
    \begin{align}
        &\1{\Pe(x) \neq \varnothing, \pi_e=a_0} \E[Y \mid X = x, \Pi \in \Pe(x)] \nonumber \\
        &= \1{\Pe(x) \neq \varnothing, \pi_e=a_0} \E[Y \mid X = x, A = \pi_e(x), \Pi \in \Pe(x)] \label{eq:inequality_bounds_proof_1}\\
        &= \1{\Pe(x) \neq \varnothing, \pi_e=a_0} \E_{M}[\E[Y \mid X = x, A = \pi_e(x), \Pi \in \Pe(x), M] \mid X = x, A = \pi_e(x), \Pi \in \Pe(x)] \label{eq:inequality_bounds_proof_2}\\
        &= \1{\Pe(x) \neq \varnothing, \pi_e=a_0} \E_{M}[\E[Y \mid X = x, A = \pi_e(x), M] \mid X = x, A = \pi_e(x), \Pi \in \Pe(x)] \label{eq:inequality_bounds_proof_3} \\
        &= \1{\Pe(x) \neq \varnothing, \pi_e=a_0} \E_{M}[\E[Y \mid X = x, A = a_0, M] \mid X = x, A = \pi_e(x), \Pi \in \Pe(x)] \label{eq:inequality_bounds_proof_4} \\
        &= \1{\Pe(x) \neq \varnothing, \pi_e=a_0} \E_{M}[\E[Y \mid X = x, A = a_0] \mid X = x, A = \pi_e(x), \Pi \in \Pe(x)] \label{eq:inequality_bounds_proof_5} \\
        &= \1{\Pe(x) \neq \varnothing, \pi_e=a_0} \E[Y \mid X = x, A = a_0] \label{eq:inequality_bounds_proof_6}\\
        &= \1{\Pe(x) \neq \varnothing, \pi_e=a_0} \E[Y \mid X = x, A = \pi_e, M = f_M(\pi_e)] \label{eq:inequality_bounds_proof_7}
    \end{align}
    where~\cref{eq:inequality_bounds_proof_1} follows from the fact that the event $\{\Pi \in \Pe(x)\} \implies \{ A = \pi_e(x)\}$ by the definition of $\Pe(x)$ (see~\cref{def:partitions}). \Cref{eq:inequality_bounds_proof_2} simply applies the law of total probability, including $M$ in the inner expectation (and where we use $\E_M$ as helpful shorthand to remind the reader that the outer expectation is taken over $M$). \Cref{eq:inequality_bounds_proof_3} uses the fact that $Y \ci \Pi \mid A, M, X$ in our data-generating process to remove $\Pi$ from the inner expectation. \Cref{eq:inequality_bounds_proof_4} replaces $\pi_e(x)$ with $a_0$ due to the outside indicator that restricts to $x$ where $\pi_e(x) = a_0$, so whenever this expression is non-zero, then $\pi_e(x) = a_0$.  Finally,~\cref{eq:inequality_bounds_proof_5} uses~\cref{lemma:independence_under_neutral}, which implies that $\E[Y \mid X = x, A = a_0, M = m] = \E[Y \mid X = x, A = a_0]$, and~\cref{eq:inequality_bounds_proof_6} uses the fact that the inner expectation is a constant value, to remove the outer expectation over $M$. From~\cref{eq:inequality_bounds_proof_6}, we can simply add back the conditioning on $M = f_M(\pi_e)$, again using~\cref{lemma:independence_under_neutral}, and recall that $\pi_e = a_0$ under the indicator, to arrive at~\cref{eq:inequality_bounds_proof_7}, which completes the proof.
\end{proof}

We will also make use of the following inequalities that follow from boundedness of $Y$ under~\cref{asmp:bounded_Y}.

\begin{restatable}[Boundedness]{lemma}{BoundednessCorr}\label{lemma:bounded_outcomes}
    Under~\cref{asmp:bounded_Y}, the following inequalities hold by the fact that $Y \in [Y_{\text{min}}, Y_{\text{max}}]$.
    \begin{align*}
        \1{\Petl(X) = \varnothing} \E[Y \mid X = x, A=\pi_e, M=f_M(\pi_e)] &\geq \1{\Petl(X) = \varnothing} Y_{\text{min}} \\
        \1{\Petg(X) = \varnothing} \E[Y \mid X = x, A=\pi_e, M=f_M(\pi_e)] &\leq \1{\Petg(X) = \varnothing} Y_{\text{max}} \\
        \1{\Pe(X) = \varnothing, \pi_e(x) = a_0} \E[Y \mid X = x, A=a_0, M=f_M(\pi_e)] &\geq \1{\Pe(X) = \varnothing, \pi_e(x) = a_0} Y_{\text{min}} \\
        \1{\Pe(X) = \varnothing, \pi_e(x) = a_0} \E[Y \mid X = x, A=a_0, M=f_M(\pi_e)] &\leq \1{\Pe(X) = \varnothing, \pi_e(x) = a_0} Y_{\text{max}}
    \end{align*}
\end{restatable}
\begin{proof}
    Each claim is immediate from the fact that $Y$ is bounded between $Y_{\text{min}}$ and $Y_{\text{max}}$, with the additional observation that for each inequality, the indicators are identical on either side.
\end{proof}

We are now prepared to prove our main result.
\BoundsWithControl*
\begin{proof}
     We begin with the lower bound, and note that the upper bound follows similarly.  
     \paragraph{Lower Bound} First, we observe that the given set indicators form a partition over $\cX$ (a set of disjoint subsets of $\cX$ whose union is equal to $\cX$), such that 
    \begin{align}
        1 &= \1{\Petl(X) \neq \varnothing, \pi_e(X) \neq a_0} + \1{\Petl(X) = \varnothing, \pi_e(X) \neq a_0} \nonumber \\
          &\quad + \1{\Pe(X) \neq \varnothing, \pi_e(X) = a_0} + \1{\Pe(X) = \varnothing, \pi_e(X) = a_0}\label{eq:lb_partition_equals_one}
    \end{align}
    Here, the subsets of $X \in \cX$ satisfying each of $\{\Petl(X) \neq \varnothing\}$ and $\{\Petl(X) = \varnothing\}$ form a partition over $\cX$ by definition as does $\{\Pe(X) \neq \varnothing\}$ and $\{\Pe(X) = \varnothing\}$. Hence, we can write that 
    \begin{align}
        &\E[Y(A=\pi_e, M=f_M(\pi_e))]  \nonumber \\
        &\quad = \E[\E[Y(A=\pi_e, M=f_M(\pi_e)) \mid X]] \label{eq:law_of_iterated_expectation}\\
        &\quad = \E[\E[Y \mid X, A=\pi_e, M=f_M(\pi_e))]] \label{eq:thm_correctness_ignorability} \\
        &\quad = \E[(\1{\Petl(X) \neq \varnothing, \pi_e(X) \neq a_0} + \1{\Petl(X) = \varnothing, \pi_e(X) \neq a_0} \nonumber \\
        &\qquad \qquad + \1{\Pe(X) \neq \varnothing, \pi_e(X) = a_0} + \1{\Pe(X) = \varnothing, \pi_e(X) = a_0})\E[Y \mid X, A=\pi_e, M=f_M(\pi_e)]] \label{eq:thm_correctness_multiply_one} \\
        &\quad = \E\big[\1{\Petl(X) \neq \varnothing, \pi_e(X) \neq a_0} \E[Y \mid X, A = \pi_e, M = f_M(\pi_e)] \nonumber \\
        &\qquad + \1{\Petl(X) = \varnothing, \pi_e(X) \neq a_0} \E[Y \mid X, A = \pi_e, M = f_M(\pi_e)] \nonumber \\
        &\qquad + \1{\Pe(X) \neq \varnothing, \pi_e(X) = a_0} \E[Y \mid X, A = \pi_e, M = f_M(\pi_e)] \nonumber \\
        &\qquad + \1{\Pe(X) = \varnothing, \pi_e(X) = a_0} \E[Y \mid X, A = \pi_e, M = f_M(\pi_e)]\big] \label{eq:thm_correctness_distribute} 
    \end{align}
    where~\cref{eq:law_of_iterated_expectation} follows from the law of iterated expectation,~\cref{eq:thm_correctness_ignorability} follows from consistency (\cref{consistency}) and the fact that $Y(a, m) \ci A, M \mid X$ (\cref{rule_2}), and~\cref{eq:thm_correctness_multiply_one} follows from~\cref{eq:lb_partition_equals_one}.  After distributing terms in~\cref{eq:thm_correctness_distribute}, the lower bound follows from the application of~\cref{lemma:bound_lemma,lemma:bounded_outcomes} to each term.

    First, $\1{\Petl(X) \neq \varnothing, \pi_e(X) \neq a_0} \E[Y \mid X, A = \pi_e, M = f_M(\pi_e)] \geq \1{\Petl(X) \neq \varnothing, \pi_e(X) \neq a_0} \E[Y \mid X, \Pi \in \Petl(x)]$ follows from \cref{lemma:bound_lemma}.

    Next, $\1{\Petl(X) = \varnothing, \pi_e(X) \neq a_0} \E[Y \mid X, A = \pi_e, M = f_M(\pi_e)] \geq \1{\Petl(X) = \varnothing, \pi_e(X) \neq a_0} Y_{min}$ follows from \cref{lemma:bounded_outcomes}.

    Next, $\1{\Pe(X) \neq \varnothing, \pi_e(X) = a_0} \E[Y \mid X, A = \pi_e, M = f_M(\pi_e)] = \1{\Pe(X) \neq \varnothing, \pi_e(X) = a_0} \E[Y \mid X, \Pi \in \Pe(x)]$ follows from \cref{lemma:bound_lemma}.

    Finally, $\1{\Pe(X) = \varnothing, \pi_e(X) = a_0} \E[Y \mid X, A = \pi_e, M = f_M(\pi_e)] \geq \1{\Pe(X) = \varnothing, \pi_e(X) = a_0} Y_{min}$ follows from \cref{lemma:bounded_outcomes}.

    As each term above is either equal to or less than or equal to their respective corresponding terms, the sum of all the components above will be less than or equal to the target estimand.

    \paragraph{Upper Bound} For the upper bound, we use the partition given by 
    \begin{align}
        1 &= \1{\Petg(X) \neq \varnothing, \pi_e(X) \neq a_0} + \1{\Petg(X) = \varnothing, \pi_e(X) \neq a_0} \nonumber \\
          &\quad + \1{\Pe(X) \neq \varnothing, \pi_e(X) = a_0} + \1{\Pe(X) = \varnothing, \pi_e(X) = a_0}\label{eq:ub_partition_equals_one}
    \end{align}
    and the argument follows similarly, such that the upper bound follows from the application of~\cref{lemma:bound_lemma,lemma:bounded_outcomes} to each term.  

    First, $\1{\Petg(X) \neq \varnothing, \pi_e(X) \neq a_0} \E[Y \mid X, A = \pi_e, M = f_M(\pi_e)] \leq \1{\Petg(X) \neq \varnothing, \pi_e(X) \neq a_0} \E[Y \mid X, \Pi \in \Petg(X)]$ follows from \cref{lemma:bound_lemma}.

    Next, $\1{\Petg(X) = \varnothing, \pi_e(X) \neq a_0} \E[Y \mid X, A = \pi_e, M = f_M(\pi_e)] \leq \1{\Petg(X) = \varnothing, \pi_e(X) \neq a_0} Y_{max}$ follows from \cref{lemma:bounded_outcomes}.

    Next, $\1{\Pe(X) \neq \varnothing, \pi_e(X) = a_0} \E[Y \mid X, A = \pi_e, M = f_M(\pi_e)] = \1{\Pe(X) \neq \varnothing, \pi_e(X) = a_0} \E[Y \mid X, \Pi \in \Pe(x)]$ follows from \cref{lemma:bound_lemma}.

    Finally, $\1{\Pe(X) = \varnothing, \pi_e(X) = a_0} \E[Y \mid X, A = \pi_e, M = f_M(\pi_e)] \leq \1{\Pe(X) = \varnothing, \pi_e(X) = a_0} Y_{max}$ follows from \cref{lemma:bounded_outcomes}.
\end{proof}

\BoundDecomposition*

\begin{proof}
    We will start by observing that certain terms cancel out in the difference of the bound $U(\pi_e) - L(\pi_e)$.  We begin by recalling the definition of these bounds from~\cref{thm:bounds_with_control}, rearranged slightly.
    \begin{align}
        L(\pi_e) &= \E[\1{\Petl(X) \neq \varnothing}\1{\pi_e \neq a_0} \E[Y \mid X=x, \Pi \in \Petl(x)] \nonumber \\
        &\quad + \1{\Petl(X) = \varnothing}\1{\pi_e \neq a_0} Y_{min} \nonumber \\
        &\quad + \1{\Pe(X) \neq \varnothing}\1{\pi_e = a_0} \E[Y \mid X=x, \Pi \in \Pe(x)] \label{eq:decomp_cancel_1a} \\
        &\quad + \1{\Pe(X) = \varnothing}\1{\pi_e = a_0} Y_{min}] \label{eq:decomp_collect_1a} \\
        U(\pi_e) &= \E[\1{\Petg(X) \neq \varnothing}\1{\pi_e \neq a_0} \E[Y \mid X=x, \Pi \in \Petg(x)] \nonumber \\
        &\quad + \1{\Petg(X) = \varnothing}\1{\pi_e \neq a_0} Y_{max} \nonumber \\
        &\quad + \1{\Pe(X) \neq \varnothing}\1{\pi_e = a_0} \E[Y \mid X=x, \Pi \in \Pe(x)] \label{eq:decomp_cancel_1b} \\
        &\quad + \1{\Pe(X) = \varnothing}\1{\pi_e = a_0} Y_{max}] \label{eq:decomp_collect_1b}
    \end{align}
    By linearity of expectation, we can remove identical terms, i.e., we can observe that in the difference $U(\pi_e) - L(\pi_e)$, the terms in~\cref{eq:decomp_cancel_1a,eq:decomp_cancel_1b} cancel, leaving us with the following after collecting similar terms~\cref{eq:decomp_collect_1a,eq:decomp_collect_1b}.
    \begin{align}
        &U(\pi_e) - L(\pi_e) \nonumber \\
        &= \E[\1{\Petg(X) \neq \varnothing}\1{\pi_e \neq a_0} \E[Y \mid X=x, \Pi \in \Petg(x)] \label{eq:decomp_term_1} \\
        &\qquad - \1{\Petl(X) \neq \varnothing}\1{\pi_e \neq a_0} \E[Y \mid X=x, \Pi \in \Petl(x)] \label{eq:decomp_term_2} \\
        &\qquad + \1{\Petg(X) = \varnothing}\1{\pi_e \neq a_0} Y_{max} \label{eq:decomp_term_3} \\
        &\qquad - \1{\Petl(X) = \varnothing}\1{\pi_e \neq a_0} Y_{min} \label{eq:decomp_term_4} \\
        &\qquad + (\1{\Pe(X) = \varnothing}\1{\pi_e = a_0}) (Y_{\text{max}} - Y_{\text{min}}) ] \label{eq:decomp_term_5}
    \end{align}

    Now we will conduct two splits of indicators, to reflect finer-grained subgroups.  
    \begin{itemize}
        \item First, we note that we can partition the subset of $\cX$ satisfying $\{\Petl(X) \neq \varnothing\}$ into two further subsets: The set of $x$ where the \textit{only} agreeing policies are those with worse performance, and the set where there also exists agreeing policies with equal or greater performance. Note that, if there exists trial policies with exactly equal performance to the new policy, both $\{\Petl(X) \neq \varnothing\}$ and $\{\Petg(X) \neq \varnothing\}$ must be true. We can argue similarly for $\{\Petg(X) \neq \varnothing\}$, and write
            \begin{align*}
                \1{\Petl(X) \neq \varnothing} = \1{\Petl(X) \neq \varnothing, \Petg(X) = \varnothing} + \1{\Petl(X) \neq \varnothing, \Petg(X) \neq \varnothing} \\
                \1{\Petg(X) \neq \varnothing} = \1{\Petg(X) \neq \varnothing, \Petl(X) = \varnothing} + \1{\Petg(X) \neq \varnothing, \Petl(X) \neq \varnothing}
            \end{align*}
        \item Second, we note that we can partition the subset of $\cX$ satisfying $\{\Petl(X) = \varnothing\}$ into two further subsets: The set of $x$ where the {\it only} agreeing trial policies have greater performance, and the set where there are no agreeing trial policies at all. We can argue similarly for $\{\Petg(X) = \varnothing\}$, and write
            \begin{align*}
                \1{\Petl(X) = \varnothing} = \1{\Pe(X) = \varnothing} + \1{\Petl(X) = \varnothing, \Petg(X) \neq \varnothing} \\
                \1{\Petg(X) = \varnothing} = \1{\Pe(X) = \varnothing} + \1{\Petg(X) = \varnothing, \Petl(X) \neq \varnothing}
            \end{align*}
    \end{itemize}
    With these equalities in mind, we can rewrite the difference as follows by expanding terms
    \begin{align}
        &U(\pi_e) - L(\pi_e) \nonumber \\
        &= \E[
        \1{\Petg(X) \neq \varnothing, \Petl(X) \neq \varnothing}\1{\pi_e \neq a_0} \E[Y \mid X = x, \Pi \in \Petg(x)] && \text{From~\cref{eq:decomp_term_1}} \label{eq:penultimate_decomp_term_1a} \\
        &\qquad + \1{\Petg(X) \neq \varnothing, \Petl(X) = \varnothing}\1{\pi_e \neq a_0} \E[Y \mid X = x, \Pi \in \Petg(x)] && \text{From~\cref{eq:decomp_term_1}} \label{eq:penultimate_decomp_term_2a}\\
        &\qquad - \1{\Petl(X) \neq \varnothing, \Petg(X) \neq \varnothing}\1{\pi_e \neq a_0} \E[Y \mid X = x, \Pi \in \Petl(x)] && \text{From~\cref{eq:decomp_term_2}} \label{eq:penultimate_decomp_term_1b}\\
        &\qquad - \1{\Petl(X) \neq \varnothing, \Petg(X) = \varnothing}\1{\pi_e \neq a_0} \E[Y \mid X = x, \Pi \in \Petl(x)] && \text{From~\cref{eq:decomp_term_2}} \label{eq:penultimate_decomp_term_3a} \\
        &\qquad + \1{\Petg(X) = \varnothing, \Petl(X) \neq \varnothing}\1{\pi_e \neq a_0} Y_{max} && \text{From~\cref{eq:decomp_term_3}} \label{eq:penultimate_decomp_term_3b} \\
        &\qquad + \1{\Pe(X) = \varnothing}\1{\pi_e \neq a_0} Y_{max} && \text{From~\cref{eq:decomp_term_3}} \label{eq:penultimate_decomp_term_4a} \\
        &\qquad - \1{\Petl(X) = \varnothing, \Petg(X) \neq \varnothing} \1{\pi_e \neq a_0} Y_{min}&& \text{From~\cref{eq:decomp_term_4}} \label{eq:penultimate_decomp_term_2b} \\
        &\qquad - \1{\Pe(X) = \varnothing} \1{\pi_e \neq a_0} Y_{min}&& \text{From~\cref{eq:decomp_term_4}} \label{eq:penultimate_decomp_term_4b} \\
        &\qquad + (\1{\Pe(X) = \varnothing}\1{\pi_e=a_0}) (Y_{\text{max}} - Y_{\text{min}}) ]&& \text{From~\cref{eq:decomp_term_5}} \label{eq:penultimate_decomp_term_5}
    \end{align}
    And by rearranging terms, we arrive at 
    \begin{align*}
        &U(\pi_e) - L(\pi_e) \\
        &= \E[\1{\Pe(X) = \varnothing} (Y_{\text{max}} - Y_{\text{min}})  && \text{\cref{eq:penultimate_decomp_term_4a,eq:penultimate_decomp_term_4b,eq:penultimate_decomp_term_5}} \\
        &\qquad + \1{\Petg(X) \neq \varnothing, \Petl(X) \neq \varnothing}\1{\pi_e \neq a_0} \\
        &\qquad \qquad (\E[Y \mid X = x, \Pi \in \Petg(x)] - \E[Y \mid X = x, \Pi \in \Petl(x)]) && \text{\cref{eq:penultimate_decomp_term_1a,eq:penultimate_decomp_term_1b}} \\
        &\qquad + \1{\Petg(X) = \varnothing, \Petl(X) \neq \varnothing}\1{\pi_e \neq a_0} (Y_{\text{max}} - \E[Y \mid X = x, \Pi \in \Petl(x)]) && \text{\cref{eq:penultimate_decomp_term_3a,eq:penultimate_decomp_term_3b}} \\
        &\qquad + \1{\Petg(X) \neq \varnothing, \Petl(X) = \varnothing}\1{\pi_e \neq a_0} (\E[Y \mid X = x, \Pi \in \Petg(x)] - Y_{\text{min}})] && \text{\cref{eq:penultimate_decomp_term_2a,eq:penultimate_decomp_term_2b}}
    \end{align*}
    After rearranging terms slightly, this gives us the desired result of the form of $\delta(X, Y, \Pi)$ stated in the theorem. The fact that $\delta(X, Y, \Pi) \geq 0$ almost surely follows from the fact that $Y_{\text{max}} \geq \E[Y \mid C] \geq Y_{\text{min}}$ for any conditioning set $C$ according to~\cref{asmp:bounded_Y}, and the fact that $(\E[Y \mid X = x, \Pi \in \Petg(x)] - \E[Y \mid X = x, \Pi \in \Petl(x)])$ is nonnegative by~\cref{asmp:more_reliable_better}.
\end{proof}

\NonIdBounds*

\begin{proof}
    Recall that we make use of the more stringent sets of comparison arms
    \begin{align*}
        \Petl(x) &= \argmax_{\pi \in \Pel(x)} f_M(\pi) & \Petg(x) &= \argmin_{\pi \in \Peg(x)} f_M(\pi).
    \end{align*}
    From here, we will construct a pair of SCMs $\cM_L, \cM_U$ that satisfy our criteria, which are defined as follows, using $f^{L}$ to denote the structural equations under $\cM_L$ and $f^{U}$ to denote the structural equations under $\cM_U$, and $f$ (without a superscript) is used to denote structural equations that are shared between the two.

    \textbf{Shared Structural Equations for $\cM_L, \cM_U$}:  Both SCMs share the following equations that respect \cref{asmp:more_reliable_better,asmp:A_control_value}, which give rise to a shared distribution over $P(X, A, M, \Pi, D)$, and these can be chosen to match any such observed distribution.
    \begin{align*}
        D &= f_D(\epsilon_D), & X &= f_X(\epsilon_X) & \Pi &= \pi_D \\
        M &= f_M(\Pi), &  A &= \Pi(X) 
    \end{align*}

    \newcommand{\fyl}{f_Y^L}
    \newcommand{\fyu}{f_Y^U}
    \newcommand{\gyl}{g_Y^L}
    \newcommand{\gyu}{g_Y^U}
    \newcommand{\Sa}{{\bf \Pi}^a}
    \newcommand{\Ma}{\mathbf{M}_{a}}
    \newcommand{\mprev}{h_{\text{prev}}}
    \newcommand{\mnext}{h_{\text{next}}}
    \textbf{Differences between $\cM_L, \cM_U$}:  These SCMs will differ in terms of how $Y$ is generated.  Let $\fyl(A, X, M, \epsilon_Y)$, $\fyu(A, X, M, \epsilon_Y)$ denote the structural equations for $\cM_L, \cM_U$ respectively.  We will define these functions for every possible set of inputs, using knowledge of the true conditional distribution $P(Y \mid X = x, A = a, M = m)$ wherever this combination of inputs $(x, a, m)$ has positive density under the observed distribution $P(x, a, m) > 0$.

    We will define these functions constructively, by defining their behavior for every value of $(x, a, m)$.  
    \begin{enumerate}
        \item First, fix any value of $x \in \cX$. For this value of $x$, we need to define the value of $\fyl, \fyu$ for all values of $a \in \cA, m \in \R$.  To do so, let 
            \begin{equation*}
                \Sa(x) \coloneqq \{\pi \in \Pi \mid \pi(x) = a\}
            \end{equation*}
            be the set of all policies (possibly an empty set) that have output $a$ on the input $x$, and let $\Ma(x) = (m_1, \ldots, m_K)$ be the (ordered) set of performance values for these policies, where $m_1$ denotes the worst performance $f_M(\pi_i)$ (breaking ties arbitrarily) of all policies $\pi_i$ where $\pi_i \in \Sa(x)$, and $m_K$ denotes the best performance, where $K$ is the size of the set $\Sa(x)$.
        \item Now we consider any arbitrary $a \in \cA$, in addition to our fixed $x \in \cX$.  Here, there are two cases to consider:
        \begin{itemize}
            \item 
                If $\Sa(x)$ is empty for this $a$, then for all $m \in \R$, we let $\fyl(x, a, m, \epsilon_Y) = Y_{\text{min}}$, $\fyu(x, a, m, \epsilon_Y) = Y_{\text{max}}$.  In other words, at this point in $x$, if no trialed policy takes action $a$, we assume the worst (for $\fyl$) and the best (for $\fyu$) possible outcomes.  We can easily verify that both $\fyl, \fyu$ satisfy~\cref{asmp:more_reliable_better,asmp:A_control_value} since the output is a constant for each function, and satisfy~\cref{asmp:bounded_Y} since the output remains bounded.
            \item 
                If $\Sa(x)$ (and consequently $\Ma(x)$) is non-empty, then we first define the behavior of $\fyl, \fyu$ at all the (observable) performance values in $\Ma(x)$ to match the conditional distribution $P(Y \mid x, a, m)$.
                \begin{align*}
                    \fyu(x, a, m, \epsilon_Y) = \fyl(x, a, m, \epsilon_Y) \sim P(Y \mid x, a, m), \forall m \in \Ma(x),
                \end{align*}
                with the additional constraint that $\fyl, \fyu$ are constant with respect to $m$ when $a = a_0$, which itself must be achievable since we assume that the true SCM generating $P$ adheres to this constraint by~\cref{asmp:A_control_value}. Note that we can always achieve the equivalence of distribution shown above by taking $\epsilon_Y$ to be a uniform random variable in $[0, 1]$, and defining our function as sampling from $P(Y \mid x, a, m)$ using the inverse CDF $f_Y(x, a, m, \epsilon_Y) = F_{Y \mid x, a, m}^{-1}(\epsilon_Y)$. Because we assume that $P(Y \mid x, a, m)$ does not violate our assumptions, it should be clear that $\fyl(x, a, m, \epsilon_Y), \fyu(x, a, m, \epsilon_Y)$ do not violate our assumptions for values of $m \in \Ma(x)$. In addition, we have that for any $m \not\in \cM(x)$, our construction above does not violate~\cref{asmp:A_control_value} (since in this case, $\fyl(x, a_0, m, \epsilon_Y)$ is constant for all values of $m$).

                We have now defined the behavior of $\fyl, \fyu$ when $a = a_0$, and when $a \neq a_0, m \in \Ma(x)$.
                Now it remains to define the behavior of $\fyl, \fyu$ for other values of $m$ when $a \neq a_0$.  For any value $m' \not\in \Ma(x)$, there are three possible scenarios: It is smaller than the smallest value ($m_1$), larger than the largest value $m_K$, or in-between two values, which we denote $\mprev(m') < m' < \mnext(m')$ without loss of generality, where $\mprev(m') = \max(m \in \Ma(x) \mid m < m')$ and $\mnext(m') = \min(m \in \Ma \mid m > m')$.  Here, $\mprev(m')$ corresponds to the performance of the \enquote{next-worst} policy among those deployed, and $\mnext(m')$ corresponds to the performance of the \enquote{next-best} policy. We define behavior on these sets as follows
                \begin{align*}
                    \fyl(x, a, m', \epsilon_Y) &= \begin{cases}
                        Y_{\text{min}}, &\ \text{if } m' < \min(\Ma(x)) \\
                        \fyl(x, a, \mprev(m'), \epsilon_Y) &\ \text{if } m' > \min(\Ma(x)), m \not\in \Ma(x)
                    \end{cases}\\
                    \fyu(x, a, m', \epsilon_Y) &= \begin{cases}
                        Y_{\text{max}}, &\ \text{if } m' > \max(\Ma(x)) \\
                        \fyu(x, a, \mnext(m'), \epsilon_Y) &\ \text{if } m' < \max(\Ma(x)), m \not\in \Ma(x)
                    \end{cases}
                \end{align*}
                In words, we have \enquote{filled in} the missing gaps in $\fyl,\fyu$ for all values of $m$ using piecewise constant functions: For any $m' \not\in \Ma(x)$, if $m'$ is worse than any observed performance, we assume the worst-case for the lower bound, and if $m'$ is better than any observed performance, we assume the best-case for the upper bound.  Otherwise, we have $\mprev(m') < m'$ and/or $m' < \mnext(m')$, and we assume for the lower bound that the outcomes at $m'$ match those at $\mprev(m')$, and for the upper bound we assume the outcomes at $m'$ match that at $\mnext(m')$.  Because we have maintained monotonicity with respect to $m$, our construction continues to satisfy our core assumptions.
        \end{itemize}
    \item We have now fully defined $\fyl, \fyu$, having defined these functions for any input $(x, a, m)$, and shown that they satisfy our core assumptions~\cref{asmp:more_reliable_better,asmp:A_control_value,asmp:bounded_Y}.  Putting it together, we have it that 
        \begin{align}
            \fyl(x, a, m, \epsilon_Y) &= \begin{cases}
            Y_{\text{min}}, &\ \text{if } \Sa(x) = \varnothing,\\
            Y_{\text{min}}, &\ \text{if } \Sa(x) \neq \varnothing, m < \min(\Ma(x)) \\
            F_{Y \mid x, a, \mprev(m)}^{-1}(\epsilon_Y), &\ \text{if } \Sa(x) \neq \varnothing, m > \min(\Ma(x)), m \not\in \Ma(x), \\
            F_{Y \mid x, a, m}^{-1}(\epsilon_Y), &\ \text{if } \Sa(x) \neq \varnothing, m \in \Ma(x),
        \end{cases} \label{eq:def_fyl} \\
            \fyu(x, a, m, \epsilon_Y) &= \begin{cases}
            Y_{\text{max}}, &\ \text{if } \Sa(x) = \varnothing,\\
            Y_{\text{max}}, &\ \text{if } \Sa(x) \neq \varnothing, m > \max(\Ma(x)) \\
            F_{Y \mid x, a, \mnext(m)}^{-1}(\epsilon_Y), &\ \text{if } \Sa(x) \neq \varnothing, m < \max(\Ma(x)), m \not\in \Ma(x), \\
            F_{Y \mid x, a, m}^{-1}(\epsilon_Y), &\ \text{if } \Sa(x) \neq \varnothing, m \in \Ma(x),
        \end{cases}\label{eq:def_fyu}
        \end{align}
        where $F^{-1}_{Y \mid x, a, m}$ is the inverse conditional CDF of $Y$ given $X, A, M$, derived from $P$, and where $\Sa(x) \coloneqq \{\pi \in \Pi \mid \pi(x) = a\}$ and $\Ma(x) \coloneqq \{f_M(\pi) : \pi \in \Sa(x)\}$, as defined previously above. 
    \end{enumerate}

    \paragraph{Verifying conditions} We have now defined the SCMs $\cM_L, \cM_U$, and shown that these SCMs are both consistent with our assumptions.  We will now briefly verify that both SCMs give rise to the same observed distribution $P(X, Y, A, M, \Pi, D)$, and then demonstrate that these SCMs achieve the upper and lower bounds that are given in~\cref{thm:bounds_with_control}.

    First, we have it by construction that both SCMs yield the observed distribution $P(X, Y, A, M, \Pi, D)$, so it remains to demonstrate that they agree with the observed distribution $P(Y \mid X, A, M, D, \Pi)$, which we can write equivalently as $P(Y \mid X, A, M)$, since $D, \Pi \ci Y \mid X, A, M$ under our assumed data-generating process.  Note that $P(Y \mid X, A, M, \Pi, D)$ is only well-defined for $x, a, m$ with positive density (if $X$ is continuous) or probability mass (if $X$ is discrete).  Assuming that $P(x) > 0$ for all $x \in \cX$, we have constructed $\fyl, \fyu$ to agree with $P(Y \mid x, a, m)$ for all $a, m$ where there exists a trialed policy $\pi$ that outputs $a = \pi(x)$ with performance $m = f_M(\pi)$.  We note that for any other value of $a', m'$, we have it that $P(a', m' \mid x) = 0$, and hence the entire set $(x, a', m')$ has zero density, and it is precisely on these never-observed subsets of inputs where $\cM_L, \cM_U$ disagree.

    Second, we can verify that the policy values under $f_Y^L$ and $f_Y^U$ evaluate to $L(\pi_e)$ and $U(\pi_e)$, respectively. Recalling that $Y(A = \pi_e, M=f_M(\pi_e)) = f_Y(X, \pi_e(X), f_M(\pi_e), \epsilon_Y)$, and recalling~\cref{eq:lb_partition_equals_one}, we can write that under $\cM_L$,
    \begin{align}
        \E_{\cM_L}[Y(A=\pi_e, M=f_M(\pi_e))] &= \E\big[
        \1{\Petl(X) \neq \varnothing, \pi_e(X) \neq a_0} \fyl(X, \pi_e(X), f_M(\pi_e), \epsilon_Y)\label{eq:tightness_lb_component_2}\\
        &\qquad + \1{\Petl(X) = \varnothing, \pi_e(X) \neq a_0} \fyl(X, \pi_e(X), f_M(\pi_e), \epsilon_Y)\label{eq:tightness_lb_component_3}\\
        &\qquad + \1{\Pe(X) \neq \varnothing, \pi_e(X) = a_0} \fyl(X, \pi_e(X), f_M(\pi_e), \epsilon_Y)\label{eq:tightness_lb_component_4}\\
        &\qquad + \1{\Pe(X) = \varnothing, \pi_e(X) = a_0} \fyl(X, \pi_e(X), f_M(\pi_e), \epsilon_Y)\label{eq:tightness_lb_component_5}]
    \end{align}
    where we can consider each component in the sum individually by linearity of expectation, and the fact that $\E[\1{x \in \Omega}f(x, \epsilon_Y)] = \E[\1{x \in \Omega}\E[f(x, \epsilon_Y) \mid x \in \Omega]]$ for any set $\Omega$.  We consider each term under the definition of $\fyl$ in~\cref{eq:def_fyl}. 
    \begin{itemize}
        \item For~\cref{eq:tightness_lb_component_2}, we can observe that for all $x \in \cX$ satisfying $\Petl(X) \neq \varnothing, a = \pi_e(x)$, the set $\Sa(x)$ is non-empty by definition of $\Petl(X)$, and moreover that $m_e \geq \min(\Ma(x))$.  As a result, we have it that $\fyl(x, a, m_e, \epsilon_Y) = F^{-1}_{Y \mid x, a, \mprev(m_e)}(\epsilon_Y) \sim P(Y \mid x, a, \mprev(m_e))$, and thus this term can be re-written as 
            \begin{equation}\label{eq:tightness_lb_collect_2}
                \1{\Petl(X) \neq \varnothing, \pi_e(X) \neq a_0} \E[Y \mid X, \Pi \in \Petl(X)] 
            \end{equation}
            since for any $\pi_i \in \Petl(X)$, we have it that $\pi_i(X) = \pi_e(X)$ and $f_M(\pi_i) = \mprev(m_e)$ by definition of $\mprev(m_e)$ and $\Petl(X)$.
        \item For~\cref{eq:tightness_lb_component_3}, we can observe that for all $x \in \cX$ satisfying $\Petl(X) = \varnothing, a = \pi_e(x)$, the set $\Sa(x)$ is either empty, or non-empty where $f_M(\pi_e) < \min(\Ma(x))$, by definition of $\Petl(X)$.  In either case, we have it that $\fyl = Y_{\text{min}}$, and so this term is equal to 
            \begin{equation}\label{eq:tightness_lb_collect_3}
                \1{\Petl(X) = \varnothing, \pi_e(X) \neq a_0} Y_{\text{min}}
            \end{equation}
        \item For~\cref{eq:tightness_lb_component_4}, we can observe that for all $x \in \cX$ satisfying $\Pe(X) \neq \varnothing, a = \pi_e(x) = a_0$, the set $\Sa(x)$ is non-empty by definition of $\Pe(X)$, and moreover that $\fyl$ is invariant to the choice of $m$, and so 
        this term is equal to 
        $\1{\Pe(X) \neq \varnothing, \pi_e(X) = a_0} \E[Y \mid X, A = \pi_e(X), M = f_M(\pi_e)]$, which is equal to 
            \begin{equation}\label{eq:tightness_lb_collect_4}
                 \1{\Pe(X) \neq \varnothing, \pi_e(X) = a_0} \E[Y \mid X, \Pi \in \Pe(X)] 
            \end{equation}
        from \cref{eq:inequality_bounds_outcome_neutral} of \cref{lemma:bound_lemma}
        \item For~\cref{eq:tightness_lb_component_5}, we can observe that for all $x \in \cX$ satisfying $\Pe(X) = \varnothing$, the set $\Sa(x)$ is empty by definition of $\Pe(X)$, and so $\fyl = Y_{\text{min}}$. Thus, this term is equal to
            \begin{equation}\label{eq:tightness_lb_collect_5}
                \1{\Pe(X) = \varnothing, \pi_e(X) = a_0} Y_{\text{min}}
            \end{equation}
    \end{itemize}
    Collecting terms~\cref{eq:tightness_lb_collect_2,eq:tightness_lb_collect_3,eq:tightness_lb_collect_4,eq:tightness_lb_collect_5} gives us that 
    \begin{align*}
        &\E_{\cM_L}[Y(A=\pi_e, M=f_M(\pi_e))] \nonumber \\
        &\quad = \E\big[\1{\Petl(X) \neq \varnothing, \pi_e(X) \neq a_0} \E[Y \mid X, \Pi \in \Petl(X)] \\
        &\qquad + \1{\Petl(X) = \varnothing, \pi_e(X) \neq a_0} Y_{\text{min}} \\
        &\qquad + \1{\Pe(X) \neq \varnothing, \pi_e(X) = a_0} \E[Y \mid X, \Pi \in \Pe(X)]\\
        &\qquad + \1{\Pe(X) = \varnothing, \pi_e(X) = a_0} Y_{\text{min}} \big]
    \end{align*}
    which is equivalent to $L(\pi_e)$ and completes the proof for the lower bound.  For the upper bound, the argument is similar, and roughly symmetric, but uses the partition given by~\cref{eq:ub_partition_equals_one}.

    \paragraph{Conclusion} Because it is possible to construct structural causal models that are consistent with our assumptions and that have counterfactual policy values that are exactly $L(\pi_e)$ and $U(\pi_e)$, the bounds in \cref{thm:bounds_with_control} cannot be improved without further assumptions.
\end{proof}

\EmpiricalEstimation*

\begin{proof}
    First, note that the conditional expectation of $Y$ is always finite due to~\cref{asmp:bounded_Y}.

    \newcommand{\genS}{\cS'}
    \newcommand{\genX}{\cX'}
    \newcommand{\genPol}{{\bf \Pi}'}
    \begin{lemma}\label{lemma:helper_empirical}
    Let $\genPol(X)$ be a function that maps from $\cX$ to any subset (including the empty set) of $\Pi$, and let $\genX$ be a subset of $\cX$.  If $P(\Pi \in \genPol(x)) > 0, \forall x \in \genX$, then 
    \begin{equation}
        \E\left[Y \frac{\1{\Pi \in \genPol(X)}}{P(\Pi \in \genPol(X))} \1{X \in \genX}\right] = \E[\E[Y \mid \Pi \in \genPol(X), X] \1{X \in \genX}]
    \end{equation}
    \end{lemma}
    \begin{proof}
    \begin{align}
        \E\left[Y \frac{\1{\Pi \in \genPol(X)}}{P(\Pi \in \genPol(X))} \1{X \in \genX}\right] &= \E\left[\E[Y \1{\Pi \in \genPol(X)} \mid X] \frac{\1{X \in \genX}}{P(\Pi \in \genPol(X))} \right] \\
        &= \E\left[\E[Y \mid \Pi \in \genPol(X), X] P(\Pi \in \genPol(X) \mid X) \frac{\1{X \in \genX}}{P(\Pi \in \genPol(X))} \right] \\
        &= \E[\E[Y \mid \Pi \in \genPol(X), X] \1{X \in \genX}]
    \end{align}
    where the first equality is well-defined on both sides by the assumption that for any $X \in \genX$, $P(\Pi \in \genPol(X)) > 0$. For the second-to-last line, note that this follows from the basic fact that $A, B, C$
    \begin{align*}
    \E[A \cdot \1{B \in \mathcal{B}} \mid C] &= \E[A \cdot \1{B \in \mathcal{B}} \mid B \in \mathcal{B}, C] P(B \in \mathcal{B} \mid C) \\
    &\quad + \E[A \cdot \1{B \in \mathcal{B}} \mid B \not\in \mathcal{B}, C] P(B \not\in \mathcal{B} \mid C) \\
    &= \E[A \mid B \in \mathcal{B}, C] P(B \in \mathcal{B} \mid C) 
    \end{align*}
    and the last line follows from the fact that $\Pi \ci X$ under~\cref{asmp:scm}, so that $P(\Pi \in \genPol(X) \mid X) = P(\Pi \in \genPol(X))$.
    \end{proof}

    Note that~\cref{lemma:helper_empirical} applies to all of the pairs (e.g., $\Pe(X) \neq \varnothing$ and $\{\Petg(X) \neq \varnothing, \pi_e(X) \neq a_0\}$) used in~\cref{cor:empirical_estimation}. Thus, we can directly write the following through linearity of expectations and two applications of \cref{lemma:helper_empirical}.
    
    \begin{align*}
        L(\pi_e) &= \E[\1{\Petl(X) \neq \varnothing}\1{\pi_e \neq a_0} \E[Y \mid X, \Pi \in \Petl(X)] \\
        &\quad + \1{\Petl(X) = \varnothing}\1{\pi_e \neq a_0} Y_{min} \\
        &\quad + \1{\Pe(X) \neq \varnothing}\1{\pi_e = a_0} \E[Y \mid X, \Pi \in \Pe(X)] \\
        &\quad + \1{\Pe(X) = \varnothing}\1{\pi_e = a_0} Y_{min}] \\
        &= \E[Y \frac{\1{\Pi \in \Petl(X)}}{P(\Pi \in \Petl(X))} \1{\Petl(X) \neq \varnothing}\1{\pi_e \neq a_0} \\
        &\quad + Y_{min} \1{\Petl(X) = \varnothing}\1{\pi_e \neq a_0} \\
        &\quad + \frac{\1{\Pi \in \Pe(X)}}{P(\Pi \in \Pe(X))} \1{\Pe(X) \neq \varnothing}\1{\pi_e = a_0} \\
        &\quad + Y_{min} \1{\Pe(X) = \varnothing}\1{\pi_e = a_0}]
    \end{align*}
    We apply~\cref{lemma:helper_empirical} in two instances: one where the indicator function $\1{\Petl(X) \neq \varnothing}$ is turned on and one where the indicator function $\1{\Pe(X) \neq \varnothing}$ is turned on. In the case where $\{\Petl(X) \neq \varnothing\}$ is true, $P(\Pi \in \Petl(X)) > 0$ is also true because there is at least one trial policy in the set $\Petl(X)$. Similarly, when $\{\Pe(X) \neq \varnothing\}$ is satisfied, $P(\Pi \in \Pe(X)) > 0$ will also be satisfied. Thus, when applying~\cref{lemma:helper_empirical}, the assumption required in the lemma that $P(\Pi \in \genPol(x)) > 0$ is satisfied by the decomposition of the lower bound, and we do not require any additional assumptions regarding the probability of deploying a particular set of trial models.
    
    Because the sets $(\Petl(X) \neq \varnothing, \pi_e \neq a_0)$, $(\Petl(X) = \varnothing, \pi_e \neq a_0)$, $(\Pe(X) \neq \varnothing, \pi_e = a_0)$, and $(\Pe(X) = \varnothing, \pi_e = a_0)$ are disjoint, only one product of the indicator functions inside the expectation above ever evaluates to $1$. Thus, we can equivalently express the expression inside the expectation above as the piecewise function
    \begin{align*}
        &\psi_L(Y, X, \Pi) \\
        &\coloneqq \begin{cases}
            Y \cdot \frac{\1{\Pi \in \Petl(X)}}{P(\Pi \in \Petl(X))}, &\ \text{if } \Petl(X) \neq \varnothing, \pi_e(X) \neq a_0 \\
            Y_{\text{min}}, &\ \text{if } \Petl(X) = \varnothing, \pi_e(X) \neq a_0 \\
            Y \cdot \frac{\1{\Pi \in \Pe(X)}}{P(\Pi \in \Pe(X))}, &\ \text{if } \Pe(X) \neq \varnothing, \pi_e(X) = a_0 \\
            Y_{\text{min}}, &\ \text{if } \Pe(X) = \varnothing, \pi_e(X) = a_0
        \end{cases}
    \end{align*}
    Thus, $L(\pi_e) = \E[\psi_L(Y, X, \Pi)]$.

    The proof for $U(\pi_e)$ follows similarly. We directly write the following through linearity of expectations and two applications of \cref{lemma:helper_empirical}.
    \begin{align*}
        U(\pi_e) &= \E[\1{\Petg(X) \neq \varnothing}\1{\pi_e \neq a_0} \E[Y \mid X, \Pi \in \Petg(X)] \\
        &\quad + \1{\Petg(X) = \varnothing}\1{\pi_e \neq a_0} Y_{max} \\
        &\quad + \1{\Pe(X) \neq \varnothing}\1{\pi_e = a_0} \E[Y \mid X, \Pi \in \Pe(X)] \\
        &\quad + \1{\Pe(X) = \varnothing}\1{\pi_e = a_0} Y_{max}] \\
        &= \E[Y \frac{\1{\Pi \in \Petg(X)}}{P(\Pi \in \Petg(X))} \1{\Petg(X) \neq \varnothing}\1{\pi_e \neq a_0} \\
        &\quad + Y_{max} \1{\Petg(X) = \varnothing}\1{\pi_e \neq a_0} \\
        &\quad + \frac{\1{\Pi \in \Pe(X)}}{P(\Pi \in \Pe(X))} \1{\Pe(X) \neq \varnothing}\1{\pi_e = a_0} \\
        &\quad + Y_{max} \1{\Pe(X) = \varnothing}\1{\pi_e = a_0}]
    \end{align*}
    Because this different enumeration of the sets $(\Petg(X) \neq \varnothing, \pi_e \neq a_0)$, $(\Petg(X) = \varnothing, \pi_e \neq a_0)$, $(\Pe(X) \neq \varnothing, \pi_e = a_0)$, and $(\Pe(X) = \varnothing, \pi_e = a_0)$ is also disjoint, only one product of the indicator functions inside the expectation above ever evaluates to $1$. Thus, we can equivalently express the expression inside the expectation above as the piecewise function
    \begin{align*}
        &\psi_U(Y, X, \Pi) \\
        &\coloneqq \begin{cases}
            Y \cdot \frac{\1{\Pi \in \Petg(X)}}{P(\Pi \in \Petg(X))}, &\ \text{if } \Petg(X) \neq \varnothing, \pi_e(X) \neq a_0 \\
            Y_{\text{max}}, &\ \text{if } \Petg(X) = \varnothing, \pi_e(X) \neq a_0 \\
            Y \cdot \frac{\1{\Pi \in \Pe(X)}}{P(\Pi \in \Pe(X))}, &\ \text{if } \Pe(X) \neq \varnothing, \pi_e(X) = a_0 \\
            Y_{\text{max}}, &\ \text{if } \Pe(X) = \varnothing, \pi_e(X) = a_0
        \end{cases}
    \end{align*}
    Thus, $U(\pi_e) = \E[\psi_U(Y, X, \Pi)]$.

    To show asymptotic normality, it suffices to observe that $\psi_U, \psi_L$ are known functions of the data, such that the problem reduces to mean estimation using samples.  The asymptotic behavior is then just a consequence of the central limit theorem \citep{Vaart_1998}, and the validity of the confidence intervals follows from the fact that we use the $1 - \alpha/2$ lower bound for $L$, such that the probability of failing to cover $L$ is asymptotically $1 - \alpha/2$, and similarly the $1 - \alpha/2$ upper bound for $U$.  The validity of the given interval follows from application of the union bound.
\end{proof}

\end{document}